\pdfoutput=1
\documentclass[11pt, letterpaper]{article}

\usepackage[margin=1in]{geometry}

\usepackage{graphicx} %
\usepackage[numbers]{natbib}

\usepackage[todo=disable,compact=false]{QuickCommands}

\title{Algorithms with Smoothed Polynomial-Time Complexity for Deterministic Discounted-sum and Mean-payoff Games}

\author{
    Ali Asadi \\
    Institute of Science and Technology Austria (ISTA) \\
    \texttt{ali.asadi@ist.ac.at}
    \and
    Krishnendu Chatterjee \\
    Institute of Science and Technology Austria (ISTA) \\
    \texttt{kchatterjee@ist.ac.at}
    \and
    Ruichen Luo \\
    Institute of Science and Technology Austria (ISTA) \\
    \texttt{rluo@ist.ac.at}
}

\date{}

\begin{document}

\maketitle
\thispagestyle{empty}

\begin{abstract}
    Turn-based deterministic games on graphs are two-player zero-sum games played on finite directed graphs,
where the vertex set is partitioned between the Max and the adversarial Min Player. The respective player chooses outgoing edges from their vertices.
We consider these games with the classical discounted-sum and mean-payoff objectives.
The goal of the Max Player is to maximize the objective against the adversarial Min Player.
These games lie in NP $\cap$ coNP and are among the rare combinatorial problems that belong to this complexity class, yet the existence of a polynomial-time algorithm is a major open question.
All known deterministic algorithms require exponential time in the worst case, while the classical strategy iteration (policy iteration) algorithm is remarkably efficient in practice, which motivates the study of the smoothed complexity of these games.
Previous results established a smoothed polynomial-time bound only for the restricted class of ergodic games and presented only a high-probability guarantee.
Our main result is a new analysis technique for a variant of the strategy iteration algorithm, and we show these algorithms achieve smoothed polynomial complexity for general game graphs with both high-probability and expected runtime guarantees.
In particular, our result removes the ergodicity restriction of prior work and strengthens the high-probability guarantee to the stronger expected-runtime guarantee, thereby resolving the open question on the smoothed complexity of these games.

\end{abstract}

\clearpage
\setcounter{page}{1}

\section{Introduction}
Two-player turn-based games on graphs with discounted-sum and mean-payoff objectives are a classical model whose polynomial-time solvability is a long-standing open problem.
In this work, we study the smoothed complexity of these games and show that variants of the policy-iteration algorithm achieve smoothed polynomial time on general game graphs, with both high-probability and expected runtime guarantees.

\paragraph{Turn-based game graphs.}
The class of two-player turn-based (perfect-information) deterministic games
is a widely studied theoretical model.
The model consists of a finite directed graph $G=(V,E)$ where the vertices
are partitioned into Player-Max and Player-Min vertices, and the corresponding
player chooses outgoing edges of the vertices in control.
This model extends classical graphs and is used to model interaction between
two adversarial agents, and has many applications, e.g.,
in analysis of reactive synthesis (interaction of controller and environment)~\cite{pnueli1989synthesis,ramadge1987supervisory};
in adversarial planning~\cite{mahanti1985and,hansen1998heuristic}.
Moreover, this model closely corresponds to the notion of alternation in Turing machines
which has been studied in~\cite{chandra1981alternation}.

\paragraph{Discounted-sum and mean-payoff objectives.}
In the study of games over graphs, the two fundamental payoff functions or objectives are
discounted-sum and mean-payoff. Every edge is assigned a real-valued reward, which represents
the per-stage reward.
A play (or infinite-walk) of the graph is an infinite-sequence of vertices, with edges between
the consecutive vertices.
The payoff of a play under discounted-sum (resp., mean-payoff) is the discounted-sum (resp., the
long-run average) of the rewards that appear in the play.
These objectives are the most well-studied and basic objectives studied in games over
graphs~\cite{shapley1953stochastic,gillette1957stochastic,filar2012competitive,puterman2014markov},
e.g., mean-payoff and discounted-sum games have been studied for quantitative reactive
synthesis~\cite{de2003discounting,chatterjee2015quantitative,bloem2009better}.

\paragraph{Algorithmic results.}
Besides the practical motivation to study turn-based deterministic games with
discounted-sum and mean-payoff objectives, they represent an intriguing class of
problems for algorithmic study.
The main computational problem is related to the optimal value for an objective,
where the Max Player aims to maximize and the Min Player aims to minimize the payoff function.
The decision problems for these games lie in NP $\cap$ coNP (also UP $\cap$ coUP)~\cite{condon1992complexity,zwick1996complexity},
but the existence of a polynomial-time algorithm is a long-standing and major open
problem.
All deterministic algorithms for this problem have a worst-case exponential-time complexity.

\paragraph{Smoothed complexity.}
While the worst-case complexity for the computational problem for these games
is exponential, a classical algorithm (namely the policy-iteration algorithm) works
very well in practice~\cite{dasdan2004experimental,kretinsky2022comparison}.
Hence it is natural to consider the smoothed complexity of this problem~\cite{spielman2004smoothed},
which is a classical complexity analysis between worst-case and average-case analysis.
The smoothed complexity analysis of these problems has an interesting history:
First, a polynomial-time smoothed complexity result for mean-payoff games was claimed in~\cite{boros2011stochastic};
unfortunately this claim was incorrect and later retracted.
Second, this result was partially recovered in~\cite{loff2024smoothed} which shows smoothed polynomial-time
complexity for the special class of ergodic games (where all vertices have the same value)
with a high-probability guarantee.

\paragraph{Open problem.}
The class of ergodic games is often simpler to analyze, e.g., partially-observable MDPs (POMDPs)
with ergodicity are decidable~\cite{chatterjee2025uniform} whereas in general they are undecidable~\cite{madani2003undecidability};
and for mean-payoff objectives the algorithmic study for ergodic games is much 
simpler~\cite{chatterjee2014complexity,hoffman1966nonterminating}.
Hence the first key open problem is to extend the smoothed polynomial time result
from ergodic to general game graphs.
Moreover, the previous result only presents a high-probability guarantee, whereas
for smoothed complexity analysis a stronger and more desired guarantee is the
expected runtime guarantee (see \cref{remark:high-probability-vs-expectation}).
Thus, the second interesting open problem is to establish polynomial-time smoothed complexity
with expected runtime guarantee for turn-based deterministic discounted-sum and
mean-payoff games.

\paragraph{Our contribution.}
We answer the above open questions in the affirmative. Moreover, we show that
variants of the policy-iteration algorithm achieve the polynomial-time smoothed complexity
on general game graphs with expected runtime guarantee.
In other words, we resolve the open question with variants of a well-known
algorithm that also works well in practice.

\paragraph{Related works.}
The algorithmic study of discounted-sum and mean-payoff games has a long history~\cite{ehrenfeucht1979positional,gurvich1988cyclic,zwick1996complexity,asadi2024deterministic,dorfman2019faster,kozachinskiy2021polyhedral}.
All these deterministic algorithms have worst-case exponential-time complexity for the general case,
and even the known randomized algorithms achieve only sub-exponential time~\cite{ludwig1995subexponential,bjorklund2007combinatorial}.
For a fixed discount factor, policy iteration is strongly polynomial for MDPs~\cite{ye2011simplex},
turn-based stochastic games~\cite{hansen2013strategy}, and robust
MDPs~\cite{asadi2026strongly,asadi2026complexity,zhong2026linear}.
Without a fixed discount factor, lower bounds for Howard's policy iteration are known already for
deterministic MDPs: for the mean-payoff objective, \cite{hansen2010lower,asadi2025lower} give lower
bounds on the number of iterations, and for the discounted-sum objective with the discount factor as
part of the input, \cite{zhong2026policy} establishes a lower bound that is exponential in the number
of states, which shows that policy iteration is not strongly polynomial in this setting.
Moreover, discounted-sum games are closely related to turn-based stochastic games,
which subsume MDPs as a special case. Thus, it is natural to ask whether our smoothed
polynomial-time result extends to this broader class.
Quite interestingly, \cite{christ2023smoothed} establishes a super-polynomial-time lower bound for
the smoothed complexity of policy iteration on MDPs, which rules out such an extension to
turn-based stochastic games.

\paragraph{Concurrent work.}
In independent work, Loff and Skomra~\cite{loff2026smoothed} study the smoothed complexity of
policy iteration and value iteration for deterministic discounted-sum and mean-payoff games on
general game graphs; we learned of their results through personal communication.

\section{Preliminaries}\label{sec:Preliminary}
For ease of reference, \cref{table:notation} collects the recurring notation of the paper. Symbols are grouped by role; entries point to the section or equation where they are first introduced.

\begin{table}[!t]
\centering
\caption{Recurring notation used throughout the paper.}
\label{table:notation}
\renewcommand{\arraystretch}{1.15}
\begin{tabular}{@{}l@{\quad}l@{}}
\hline
\multicolumn{2}{@{}l}{\textit{Game graph and games}} \\
\hline
$G = (V_{\min}, V_{\max}, E)$ & game graph; $V = V_{\min} \uplus V_{\max}$, $n = \Abs{V}$, $m = \Abs{E}$ \\
$G^{\gamma, \bfr}$ & discounted-sum game on $G$ with discount $\gamma$ and rewards $\bfr$ \\
$G^{\text{mp}, \bfr}$ & mean-payoff game on $G$ with rewards $\bfr$ \\
\hline
\multicolumn{2}{@{}l}{\textit{Strategies}} \\
\hline
$\pi_{\min}, \pi_{\max}$ & positional strategies of the Min / Max player \\
$\Pi_{\min}, \Pi_{\max}$ & sets of all positional strategies \\
$\pi^{\text{bw}}_{\min}, \pi^{\text{bw}}_{\max}$ & Blackwell optimal strategies (\cref{sec:background}) \\
$I_{\min}, I_{\max}$ & vertices with improving switches for Min / Max \\
\hline
\multicolumn{2}{@{}l}{\textit{Rewards, perturbation, values}} \\
\hline
$\bfr_0 \in [-1, 1]^E$ & adversarial (unperturbed) reward vector \\
$\xi \in \R^E$ & noise vector; density bounded by $\phi$, tail parameter $\theta$ \\
$\bfr = \bfr_0 + \xi \in \R^E$ & perturbed reward vector (\cref{sec:smoothed-model}) \\
$\gamma \in [0, 1)$ & discount factor \\
$\mu^{\gamma, \bfr}(v), \, \mu^{\text{mp}, \bfr}(v)$ & discounted-sum / mean-payoff value at vertex $v$ \\
$\gamma_{\text{bw}}(\bfr)$ & Blackwell threshold (\cref{sec:background}) \\
\hline
\multicolumn{2}{@{}l}{\textit{Algorithms}} \\
\hline
$\algname{PI}(G, \gamma, \bfr)$ & policy iteration on $G^{\gamma, \bfr}$ (\cref{sec:background}) \\
$\algname{RePI}$ & restarting policy iteration for discounted-sum games (\cref{alg:restarting-tpi}) \\
$\algname{RePI}_{\text{mp}}$ & restarting policy iteration for mean-payoff games (\cref{alg:restarting-tpi-mp}) \\
$\gamma_k = 1 - e^{-k}$ & geometric threshold schedule used by \algname{RePI} / $\algname{RePI}_{\text{mp}}$ \\
\hline
\multicolumn{2}{@{}l}{\textit{Analysis quantities}} \\
\hline
$\Delta^{\gamma, \bfr}(v, v^\prime)$ & suboptimality gap of the edge $(v, v^\prime)$ \\
$K \triangleq 92 m^2 n^2 (\phi + 1)$ & coefficient in the tail bound, defined in \cref{eq:K-definition} \\
$M(x) \triangleq 1 + \tfrac{1}{\theta} \ln(m x)$ & tail bound function (\cref{lemma:general-bw-threshold-tail}) \\
$Y(\bfr)$ & truncated inverse gap (\cref{sec:expected-running-time-ergodic}) \\
$\epsilon \in (0, 1)$ & failure probability for high-probability guarantees \\
\hline
\end{tabular}
\end{table}

\begin{definition}[Game graphs]
A game graph $G = (V_{\min}, V_{\max}, E)$ is a directed finite graph consisting of two disjoint finite sets of vertices $V_{\min}$ and $V_{\max}$, with vertex set $V \triangleq V_{\min} \uplus V_{\max}$ and edges $E \subseteq V \times V$.
We denote by $n\triangleq \Abs {V}$ and $m \triangleq \Abs {E}$.
An edge $a = (v, v^\prime) \in E$ indicates the transition from vertex $v$ to $v^\prime$.
Without loss of generality, we consider that for every vertex $v \in V$, there exists at least one $v^\prime \in V$ such that $(v, v^\prime) \in E$.
\end{definition}

\paragraph{Dynamics and paths.}
At the beginning of the game, a token is placed at an initial vertex $v_0 \in V$.
At any time $t$~($ t \in \Z _{\ge 0}$):
if $v_t \in V_{\min}$, then the Min Player chooses an edge $a_t \in E$, where $a_t = (v_t, v_{t+1})$, and the token is moved from $v_t$ to $v_{t+1}$;
and similarly, if $v_t \in V_{\max}$, then the Max Player chooses an edge and moves the token.
A \emph{path} (or play) starting from an initial vertex $v_0$ is an infinite sequence of vertices $\rho = v_0, v_1, v_2, \dots$ such that $(v_t, v_{t+1}) \in E$ for all $t \in \Z_{\ge 0}$.

\paragraph{Objectives.} We consider \emph{discounted-sum} and \emph{mean-payoff} objectives, and in both cases
there is a reward vector $\bfr \in \R^E$ that assigns to each edge $a = (v, v^\prime) \in E$ a real number $\bfr(v, v^\prime)$ as its per-stage reward.
\begin{itemize}
    \item \emph{Discounted-sum objective:} For the discounted-sum objective, a real discount factor $\gamma \in [0, 1)$ is given, and the game is denoted by $G^{\gamma, \bfr}$.
For any path $\rho = v_0, v_1, v_2, \dots$, the discounted-sum payoff is the infinite discounted sum of the rewards of the edges that appear in the path:
\[
    u^{\gamma, \bfr}(\rho) = \sum_{t=0}^{\infty} \gamma^t \bfr(v_t, v_{t+1}) .
\]
Since  $\gamma \in [0,1)$ and the graph is finite (i.e., rewards are bounded), this infinite sum converges.
\item \emph{Mean-payoff objectives.} For mean-payoff objectives, the same graph and reward vector are considered, denoted by $G^{\text{mp}, \bfr}$, but the payoff function evaluates a path by its long-run average reward instead of discounted sum.
For any path $\rho = v_0, v_1, \dots$, the \emph{mean-payoff} is:
\[
    u^{\text{mp}, \bfr}(\rho) = \liminf_{T \to \infty} \frac{1}{T} \sum_{t=0}^{T-1} \bfr(v_t, v_{t+1}) .
\]
\end{itemize}

\paragraph{Strategies.} A \emph{strategy} (or policy) is a rule that a player uses to choose their next edge.
In the most general sense, a strategy can depend on the prefixes of paths.
However, it is known that for both discounted-sum and mean-payoff games with perfect information, memoryless and deterministic optimal strategies always exist~\citep{condon1992complexity,ehrenfeucht1979positional}.
We refer to strategies that are both memoryless and deterministic as \emph{positional} strategies, to which we restrict our attention without loss of generality.
Formally, we define a strategy for the Min Player as a mapping $\pi_{\min} : V_{\min} \to V$ such that $(v, \pi_{\min}(v)) \in E$ for all $v \in V_{\min}$.
Symmetrically, a strategy for the Max Player is a mapping $\pi_{\max} : V_{\max} \to V$ such that $(v, \pi_{\max}(v)) \in E$ for all $v \in V_{\max}$.
We denote the sets of all valid positional strategies for the Min and Max Players as $\Pi_{\min}$ and $\Pi_{\max}$, respectively.
Given an initial vertex $v \in V$ and a pair of positional strategies $(\pi_{\min}, \pi_{\max})$,
the game graph $G$ produces a unique path, which we denote by $\rho(v, \pi_{\min}, \pi_{\max})$.
Note that this path consists of a finite-path followed by a simple cycle infinitely repeated.

\paragraph{Values.}
It is well-known from the literature~\citep{condon1992complexity} that for any given $\gamma$ and $\bfr$, every vertex $v \in V$ has a unique discounted-sum \emph{value}, denoted by $\mu^{\gamma, \bfr}(v)$, which satisfies:
\[
    \mu^{\gamma, \bfr}(v) = \min_{\pi_{\min} \in \Pi_{\min}} \max_{\pi_{\max} \in \Pi_{\max}} u^{\gamma, \bfr}(\rho(v, \pi_{\min}, \pi_{\max})) = \max_{\pi_{\max} \in \Pi_{\max}} \min_{\pi_{\min} \in \Pi_{\min}} u^{\gamma, \bfr}(\rho(v, \pi_{\min}, \pi_{\max})) .
\]
Similarly, the mean-payoff game $G^{\text{mp}, \bfr}$ has a unique value for every $v \in V$, denoted by $\mu^{\text{mp}, \bfr}(v)$. %

\paragraph{Optimal strategies.} A strategy $\pi_{\min}^* \in \Pi_{\min}$ is an \emph{optimal min strategy} for the discounted-sum game $G^{\gamma, \bfr}$ if it guarantees that the discounted-sum payoff never exceeds the discounted-sum value, regardless of the opponent's strategy; i.e., for all starting vertices $v \in V$ and all $\pi_{\max} \in \Pi_{\max}$, we have $u^{\gamma, \bfr}(\rho(v, \pi_{\min}^*, \pi_{\max})) \le \mu^{\gamma, \bfr}(v)$.
The notion of optimal max strategies is defined symmetrically.
We say that $(\pi^*_{\min}, \pi^*_{\max})$ is a pair of optimal strategies for $G^{\gamma, \bfr}$ if $\pi^*_{\min}$ is an optimal min strategy and $\pi^*_{\max}$ is an optimal max strategy for the discounted-sum objective. The notion of optimal strategies for mean-payoff games is defined analogously.

\section{Smoothed model}\label{sec:smoothed-model}
The current best-known deterministic algorithms for discounted-sum and mean-payoff games have worst-case exponential complexity. However, in practice, instances of these games rarely exhibit worst-case behaviors.
Hence the goal of this work is to study the problems through the lens of \emph{smoothed analysis}~\citep{spielman2004smoothed},
which bridges the gap between average-case and worst-case complexity by analyzing the performance of an algorithm under slight random perturbations of adversarial inputs.
Throughout this section, we describe the smoothed model for discounted-sum games. The model for mean-payoff games is obtained similarly.

\paragraph{Input formalization.} We follow the standard notion of smoothed model where the adversarial input is normalized. In discounted-sum and mean-payoff games, the rewards can be shifted and scaled to ensure that they lie in the interval $[-1,1]$.
In our smoothed analysis framework, an adversary first specifies the graph $G = (V_{\min}, V_{\max}, E)$ and the discount factor $\gamma \in (0,1)$.
The adversary also chooses an initial, worst-case reward vector $\bfr_0 \in [-1, 1]^E$.
That is, the adversary can assign any initial reward to any edge, provided it is bounded within the interval $[-1, 1]$.
After the adversary fixes $G$, $\gamma$, and $\bfr_0$, a random continuous perturbation is applied to the rewards.
Specifically, we define a noise vector $\xi \in \R^E$.
For broad application, we do not restrict $\xi$ to a specific distribution.
Instead, we require the perturbation to satisfy a set of standard regularity properties, parameterized by a maximum density bound $\phi > 0$ and a tail decay parameter $\theta > 0$.

\begin{definition}[Smoothed model]\label{def:noise-model}
The random noise vector $\xi \in \R^E$ satisfies the following properties:
\begin{enumerate}
    \item \emph{Independence:} The coordinates $\xi_a$ for all $a \in E$ are mutually independent.
    \item \emph{Density bound $\phi$:} For all $a \in E$, the probability distribution of $\xi_a$ is absolutely continuous with respect to the Lebesgue measure. Furthermore, its probability density function is uniformly bounded by $\phi$.
    \item \emph{Tail bound $\theta$:} For all $a \in E$ and $t > 0$, the tail probability satisfies $\Pr [\Abs{\xi_a} > t] \le \exp\bigl(- \theta \, t \bigr)$.
\end{enumerate}
\end{definition}

\paragraph{Special cases.}
Our general smoothed model captures many standard continuous distributions: %
\begin{itemize}
    \item \emph{Gaussian noise:} For $\xi_a \sim \cN(0, \sigma^2)$, the conditions hold with $\phi = \frac{1}{\sigma \sqrt{2\pi}}$ and $\theta = \frac{\sqrt{2/\pi}}{\sigma}$.
    \item \emph{Laplacian noise:} For $\xi_a \sim \text{Laplace}(0, b)$, the conditions hold with $\phi = \frac{1}{2b}$ and $\theta = \frac{1}{b}$.
    \item \emph{Uniform noise:} For $\xi_a \sim \text{Uniform}(-b, b)$, the conditions hold with $\phi = \frac{1}{2b}$ and $\theta = \frac{1}{b}$.
\end{itemize}

\paragraph{Problem statement.}
Given the input consisting of a game graph $G$, discount factor $\gamma$, and the normalized reward vector $\bfr_0$,
the algorithmic problem is to compute optimal strategies and values for the \emph{smoothed game} $G^{\gamma, \bfr}$, where the final perturbed reward vector $\bfr \in \R^E$ is given by
$\bfr = \bfr_0 + \xi$.

\paragraph{Probability measure.} The random noise $\xi$ induces a unique probability measure over the space of perturbed reward vectors $\bfr$. We define the probability measure as $\Pr_{\xi}[\cdot]$ and the corresponding expectation under this measure as $\mathbb{E}_{\xi}[\cdot]$.
Since for any edge $a$ the adversarial base rewards are bounded such that $\bfr_0(a) \in [-1, 1]$, the tail concentration property guarantees that the magnitude of the final perturbed reward satisfies
\begin{equation}\label{eq:reward-tail-bound}
    \Pr_{\xi}[\Abs{\bfr(a)} > 1 + t] \le \exp\bigl(- \theta \, t \bigr).
\end{equation}

\paragraph{Polynomial-time smoothed complexity.} Finally, we define the smoothed complexity of an algorithm for this problem.
Let $\cT(G, \gamma, \bfr_0 + \xi)$ denote the total number of arithmetic operations of an algorithm on the perturbed game $G^{\gamma, \bfr}$. We evaluate the algorithm's performance under two standard types of smoothed guarantees:
\begin{enumerate}
    \item \emph{High probability guarantee:} An algorithm achieves \emph{smoothed polynomial time with high probability} if, for any failure probability $\epsilon \in (0, 1)$, its runtime is bounded by a polynomial in the game dimensions ($n$ and $m$), the noise parameters ($\phi$ and $1/\theta$), and the inverse failure probability ($1/\epsilon$), with probability at least $1 - \epsilon$. Formally, for any graph $G$, discount factor $\gamma$, initial reward vector $\bfr_0 \in [-1, 1]^E$:
    \[
        \Pr_{\xi} \Bigl[ \cT(G, \gamma, \bfr_0 + \xi) \le \text{poly} \Bigl(n, m, \phi, \frac{1}{\theta}, \frac{1}{\epsilon} \Bigr) \Bigr] \ge 1 - \epsilon .
    \]

    \item \emph{Expected runtime guarantee:} An algorithm achieves \emph{expected smoothed polynomial time} if its expected runtime over the random noise is bounded by a polynomial in $n, m, \phi$, and $1/\theta$. Formally, for any graph $G$, discount factor $\gamma$, and initial reward vector $\bfr_0 \in [-1, 1]^E$:
    \[
        \mathbb{E}_{\xi} \Bigl[ \cT(G, \gamma, \bfr_0 + \xi) \Bigr] \le \text{poly} \Bigl(n, m, \phi, \frac{1}{\theta} \Bigr) .
    \]
\end{enumerate}

\begin{remark}[High probability vs. expectation guarantee]\label{remark:high-probability-vs-expectation}
We clarify the two types of guarantees above in the context of smoothed complexity.
In the design of randomized algorithms, a high-probability bound is often considered a stronger guarantee
because the algorithm can draw fresh randomness and restart upon failure.
In smoothed analysis, however, the random perturbation is fixed exactly once, making restarts impossible.
Consequently, an expected polynomial runtime is a strictly stronger mathematical guarantee.
As specifically noted in \cite{spielman2009smoothed}, a high-probability polynomial bound follows from an expected polynomial bound via Markov's inequality.
The converse, however, does not hold: an algorithm may run efficiently on $1-\epsilon$ fraction of the perturbations, yet its runtime could grow exponentially on the remaining $\epsilon$ fraction of bad instances, resulting in exponential expected runtime.
Thus, bounding the expected runtime is the desired and stronger guarantee in the context of smoothed analysis.
Thus the main goal of the study is to provide above guarantees, especially expected runtime guarantee, for
discounted-sum and mean-payoff games.
\end{remark}

\section{Overview of Techniques}\label{sec:techniques-overview}
In this section, we first recall the previous result of~\cite{loff2024smoothed} and the techniques behind it, and then give an overview of our four technical ingredients and of the novelty of our analysis.

\paragraph{Previous result and technique.}
The previous result~\cite{loff2024smoothed} in the literature for polynomial-time smoothed complexity for games
with mean-payoff and discounted-sum objectives is restrictive in two ways:
First, the class of game graphs is restricted to ergodic game graphs, where all vertices have the same value for mean-payoff objectives.
Second, the result only provides smoothed high-probability guarantee and not the more general expected guarantee.
The main algorithm for the smoothed analysis is the policy-iteration algorithm, which is simpler for ergodic mean-payoff games,
as compared to general mean-payoff games.
While for general mean-payoff games, the policy-iteration algorithm for mean-payoff objectives relies on (i)~values and
(ii)~potential or bias; for ergodic games, since all vertices have the same value, the only relevant concept is
the bias.
The results of~\cite{loff2024smoothed} depend on the notion of bias-induced optimal strategies, which is only relevant for ergodic games.
The key mathematical techniques in~\cite{loff2024smoothed} along with bias-induced optimal strategies are (i)~the notion of condition number, and
(ii)~geometric interpretation of reward vectors and hyper-plane separation induced by strategy profiles.

\paragraph{Our techniques.}
Our analysis provides the general result removing all the above restrictions. We present our key technical contributions below.
\begin{enumerate}

\item \emph{Ingredient~1:} First, we provide a quantitative decomposition of discounted-sum values with respect to mean-payoff values (\cref{lemma:value-decomposition}) and then present a Lipschitz continuity result with respect to switches of edges from optimal strategies (\cref{lemma:bounded-improvement-function}).
These results hold for all game graphs, and are neither specific to the smoothed model nor a probabilistic guarantee.

\item \emph{Ingredient~2:} Second, we present a probabilistic separation between (i)~the value of a vertex and (ii)~the value obtained by playing a sub-optimal edge and then playing optimally~(\cref{lemma:value-separation-at-one}); which we call the separation of the sub-optimality gap.

\item \emph{Ingredient~3:} Finally, for general mean-payoff games, the policy iteration algorithm depends on both the values and the bias.
Our final key ingredient is a mean-payoff value separation bound, which shows that for every vertex pair either the value is the same or
there is a separation bound for the values (\cref{lemma:non-ergodic-value-separation}).

\end{enumerate}
Ingredient~1 is the foundational result which is the basis of all results.
We remark that with Ingredient~1 and Ingredient~2 we have an alternative proof for ergodic
games with high-probability guarantee. However, for brevity, we focus on general game graphs for this paper.
Ingredient~3 along with the other two ingredients is the main idea for generalization
from ergodic to all game graphs, for high-probability guarantee.
For the extension from high-probability to expected runtime guarantee we present another technical
contribution:
\begin{enumerate}
\setcounter{enumi}{3}
\item  \emph{Ingredient~4:} We consider the Blackwell-optimal discount factor $\gamma_{\text{\upshape bw}} $ such that optimal strategies for discounted-sum are also optimal for the mean-payoff objectives.  Let the inverse gap be $1/(1-\gamma_{\text{\upshape bw}})$.
The runtime complexity of the policy-iteration algorithm is dominated by the inverse gap, and in general
the expected value of this inverse gap may be exponential.
We therefore consider a truncated version of this inverse gap, which is the minimum of the inverse gap and the number of strategies, and show a polynomial
upper bound for the expectation of the truncated version (\cref{lemma:general-expected-truncated-bw}).
\end{enumerate}

\paragraph{Novelty.}
While some individual steps of our analysis are standard in isolation, our results rely on two new technical insights.
First, removing the ergodicity assumption requires overcoming a difficulty that does not appear in~\cite{loff2024smoothed}: their analysis relies on the fact that all vertices have the same mean-payoff value, which does not hold in general game graphs.
To overcome this difficulty, we prove the mean-payoff value separation of Ingredient~3 (\cref{lemma:non-ergodic-value-separation}) with a new argument based on the reachability of recurrent cycles under Blackwell-optimal strategies (\cref{lemma:reachability-separation}).
Together with the separation of the suboptimality gap of Ingredient~2, this separation allows us to control the Blackwell threshold without the ergodicity assumption (\cref{lemma:general-bw-threshold-tail}).
Second, our truncation-and-fallback argument extends the high-probability guarantee to an expected runtime guarantee: the fallback of \algname{RePI} yields a deterministic bound on its number of iterations (\cref{lemma:general-repi-total-iteration-bound}), and hence it suffices to bound the expectation of the truncated inverse gap of Ingredient~4 (\cref{lemma:general-expected-truncated-bw}).
These two insights, combined with standard arguments, constitute the main novelty of our analysis.

Our ingredients tightly relate discounted-sum and mean-payoff games and hence we view mean-payoff
games through the lens of discounted-sum games and Blackwell-optimality.
Hence our main analysis focuses on discounted-sum games with arbitrary discount factors.
We start with the formal background on Blackwell-optimality and policy iteration algorithm for discounted-sum games.

\section{Background on Blackwell Optimality and Policy Iteration}\label{sec:background}
In this section, we provide some useful theoretical and algorithmic tools from the literature.
We first recall the concept of Blackwell optimality, which is a fundamental result in the theory of dynamic programming and games that establishes the existence of strategies, namely \emph{Blackwell optimal strategies}, that are uniformly optimal for all discount factors sufficiently close to $1$~\citep{blackwell1962discrete,bewley1976asymptotic}.
This property provides an important bridge between discounted-sum and mean-payoff games.
We then describe the classic policy iteration algorithm for discounted-sum games to compute the optimal strategies and game values.

\paragraph{Blackwell threshold $\gamma_{\text{bw}}(\bfr)$ and Blackwell-optimal strategies.}
Formally, given a specified reward vector $\bfr$, a strategy $\pi_{\min}^{\text{bw}} \in \Pi_{\min}$ is a \emph{Blackwell optimal min strategy} if there exists a threshold discount factor $\gamma_0 \in [0, 1)$ such that for all $\gamma \in [\gamma_0, 1)$, the strategy $\pi_{\min}^{\text{bw}}$ is an optimal min strategy for the discounted-sum game $G^{\gamma, \bfr}$.
Symmetrically, a strategy $\pi_{\max}^{\text{bw}} \in \Pi_{\max}$ is a \emph{Blackwell optimal max strategy} with respect to $\bfr$ if there exists a threshold discount factor $\gamma_0 \in [0, 1)$ such that for all $\gamma \in [\gamma_0, 1)$, $\pi_{\max}^{\text{bw}}$ is an optimal max strategy for $G^{\gamma, \bfr}$.
We define the \emph{Blackwell threshold} with respect to $\bfr$, denoted by $\gamma_{\text{bw}}(\bfr) \in [0, 1)$, as the infimum over all such valid thresholds.

\paragraph{Asymptotic relationship between discounted-sum and mean-payoff values.} By definition, for any $\gamma \in [\gamma_{\text{bw}}(\bfr), 1)$, a pair of Blackwell optimal strategies $(\pi^{\text{bw}}_{\min}, \pi^{\text{bw}}_{\max})$ is guaranteed to be optimal for the discounted-sum game $G^{\gamma, \bfr}$.
Furthermore, these strategies are of particular theoretical importance because they simultaneously act as optimal strategies for the corresponding mean-payoff game $G^{\text{mp}, \bfr}$, providing a crucial bridge between discounted-sum and mean-payoff games~\citep{bewley1976asymptotic}.
Formally, for any reward vector $\bfr \in \R^E$ and vertex $v \in V$, the normalized discounted-sum value converges exactly to the mean-payoff value, i.e.,
\begin{equation}\label{eq:blackwell-optimality}
    \mu^{\text{mp}, \bfr}(v) = \lim _{\gamma \uparrow 1}\ (1-\gamma) \mu^{\gamma, \bfr} (v).
\end{equation}

\paragraph{Policy iteration.} To compute the optimal strategies and values in the discounted-sum game $G^{\gamma, \bfr}$, we describe the classic policy iteration algorithm, formally presented in \cref{alg:policy-iteration}.
The algorithm operates from the perspective of one of the players; without loss of generality, we formulate it for the Min Player.
Starting from an arbitrary initial strategy, the algorithm iteratively refines the strategy $\pi_{\min}$ through alternating phases of policy evaluation and policy improvement.
\begin{itemize}
\item \emph{Policy evaluation phase:} The algorithm fixes the Min Player's current strategy $\pi_{\min}$.
With $\pi_{\min}$ fixed, the game reduces to a graph with only choices for the Max Player (aka deterministic MDP).
Graphs (or deterministic MDPs) with discounted-sum objectives can be solved in strongly polynomial time~\citep{post2015simplex}.
The resulting vector $\mu$ represents the expected discounted-sum payoff under the current strategies.
\item \emph{Policy improvement phase:} The algorithm searches for \emph{improving switches}---local edge changes that strictly decrease the current evaluated payoff for the Min Player.
If no improving switch exists for any vertex in $V_{\min}$, the algorithm terminates.
Otherwise, the algorithm updates $\pi_{\min}$ by simultaneously applying the best improving switch at every possible vertex.
This specific greedy update mechanism is known as Howard's all-switch rule.
\end{itemize}
In other words, the policy-iteration algorithm iterates over strategies of a player by greedily switching over local improvements.

\begin{algorithm}[htb]
\caption{Policy Iteration: $\algname{PI} (G, \gamma, \bfr)$}\label{alg:policy-iteration}
\begin{algorithmic}[1]
    \State \textbf{Input:} Graph $G = (V_{\min}, V_{\max}, E)$, discount factor $\gamma \in [0,1)$, reward vector $\bfr \in \R^E$.
    \State \textbf{Output:} Optimal strategy pair $(\pi_{\min}^*, \pi_{\max}^*)$ and value vector $\mu^{\gamma, \bfr}$ for the game $G^{\gamma, \bfr}$.
    \State Initialize an arbitrary positional strategy $\pi_{\min} \in \Pi_{\min}$ for the Min Player.
    \Loop
        \State \textbf{// Policy Evaluation Phase}
        \State Compute the value vector $\mu \in \R^V$ and Max Player's best-response $\pi_{\max} \in \Pi_{\max}$ satisfying:
        \Statex \qquad $\mu(v) = \bfr(v, \pi_{\min}(v)) + \gamma \mu(\pi_{\min}(v))$ for all $v \in V_{\min}$,
        \Statex \qquad $\mu(v) = \max_{(v, v^\prime) \in E} \bigl( \bfr(v, v^\prime) + \gamma \mu(v^\prime) \bigr)$ for all $v \in V_{\max}$.
        \State \textbf{// Policy Improvement Phase}
        \State Let $I = \bigl\{ v \in V_{\min} \mid \exists (v, v^\prime) \in E \text{ such that } \bfr(v, v^\prime) + \gamma \mu(v^\prime) < \mu(v) \bigr\}$.
        \If{$I = \emptyset$}
            \State \Return $(\pi_{\min}, \pi_{\max}, \mu)$
        \EndIf
        \State Update $\pi_{\min}$ by switching the edge for \textbf{all} vertices $v \in I$ to minimize the value:
        \Statex \qquad $\pi_{\min}(v) \gets \arg\min_{(v, v^\prime) \in E} \bigl( \bfr(v, v^\prime) + \gamma \mu(v^\prime) \bigr)$.
    \EndLoop
\end{algorithmic}
\end{algorithm}

\paragraph{Correctness and complexity.}
The key correctness proof establishes that when no local-improvement switches are possible, then the obtained strategy is globally optimal~\cite{condon1992complexity}.
Since the total number of strategies is at most $n^n$ and each update strictly decreases the value vector,
the algorithm is guaranteed to converge in finite time~\cite{filar2012competitive}.
Specifically, for discounted-sum games, the upper bound on the number of iterations depends on the discount factor~$\gamma$.
The worst-case strongly polynomial bound for a constant discount factor from the literature, which is the basic result we use in subsequent analysis,  is as follows:

\begin{proposition}[{\cite[Theorem 7.5]{hansen2013strategy}}]\label{prop:pi-worst-case-bound}
Let $\gamma \in [0, 1)$. For any game $G^{\gamma, \bfr}$ with $n$ vertices and $m$ edges, the number of iterations that \algname{PI} algorithm (\cref{alg:policy-iteration}) requires is at most $6 \frac{m}{1-\gamma} \ln\Bigl(\frac{n}{1-\gamma}\Bigr)$.
\end{proposition}
Each iteration of \algname{PI} performs a single policy evaluation, which requires at most $\cO (n^3 m^2 \log^2 n)$ arithmetic operations~\cite{post2015simplex}.
While this algorithm performs well in practice~\cite{dasdan2004experimental,kretinsky2022comparison}, the worst-case bound of \cref{prop:pi-worst-case-bound} is exponential as $\gamma$ approaches $1$ (e.g., $\gamma=1 -2^{-n}$).

\section{Mathematical Ingredients~1-3}\label{sec:math-properties}
In this section, we present the details related to the first three ingredients mentioned in \cref{sec:techniques-overview}.

\subsection{Ingredient~1: Value Decomposition and Lipschitz Continuity}\label{sec:ingredient-1}
In this subsection, we present conceptual ideas behind Ingredient~1 from \cref{sec:techniques-overview}.
\begin{itemize}
\item We first establish a value decomposition for discounted-sum games under Blackwell optimality (\cref{lemma:value-decomposition}), which provides a useful connection between discounted-sum and mean-payoff games.
As compared to \Cref{eq:blackwell-optimality} which only provides convergence-in-the-limit guarantee, this result provides a
quantitative bound on the difference of the discounted-sum and mean-payoff values.
While a general characterization of discounted-sum and mean-payoff value is provided in~\cite{bewley1976asymptotic} for concurrent games (with simultaneous interaction), we present a quantitative characterization albeit for the simpler class of turn-based games.

\item We then define the suboptimality gap, which quantifies the suboptimality incurred by switching a single edge in the optimal strategy.
We show that whenever the two endpoints of an edge share the same mean-payoff value, the suboptimality gap is Lipschitz continuous with respect to the discount factor $\gamma$ over the Blackwell threshold (\cref{lemma:bounded-improvement-function})
\end{itemize}

\begin{lemma}[Value decomposition under Blackwell optimality]\label{lemma:value-decomposition}
Consider a game graph $G$ with reward vector $\bfr \in \R^E$.
Let $v \in V$ be a vertex.
\begin{enumerate}
\item \emph{Decomposition:}
For all $\gamma \in (\gamma_{\text{bw}}(\bfr), 1)$, the discounted-sum value admits the following decomposition, where $w^{\bfr}_v(\gamma)$ is a rational function of $\gamma$:
\[
    \mu^{\gamma, \bfr}(v) = \frac{\mu^{\text{mp}, \bfr}(v)}{1-\gamma} + w^{\bfr}_v(\gamma) .
\]

\item \emph{Magnitude and derivative:}
For all $\gamma \in (\gamma_{\text{bw}}(\bfr), 1)$, the magnitude and the derivative of $w^{\bfr}_v(\gamma)$ are bounded by:
\[
    \Abs {w^{\bfr}_v(\gamma)} \le 4n \Norm{\bfr}_\infty \quad \text{and} \quad \Abs[\Big]{(\frac{d}{d\gamma} w^{\bfr}_v)(\gamma)} \le 7n^2 \Norm{\bfr}_\infty .
\]
\end{enumerate}
\end{lemma}

\begin{proof}[Proof overview]
Let $\mu \triangleq \mu^{\text{mp}, \bfr}(v)$.
Fix a Blackwell optimal strategy pair, under which the path from $v$ consists of a transient path of length $t \le n$ followed by a cycle of length $c \le n$.
Let $r_k$ denote the reward at step $k$. Then $\mu$ is the average cycle reward, so $\sum_{j=0}^{c-1}(r_{t+j}-\mu)=0$.
Writing each reward as $(r_k-\mu)+\mu$ in
$
    \mu^{\gamma, \bfr}(v)
    = \sum_{j=0}^{t-1} r_j \gamma^j
    + \frac{\gamma^t}{1-\gamma^c}\sum_{j=0}^{c-1}r_{t+j}\gamma^j
$
shows that the $\mu$ terms contribute $\mu/(1-\gamma)$.
The cycle cancellation removes this singular factor from the remaining terms, which define the rational function $w^{\bfr}_v(\gamma)$.
For magnitude and derivative, $w^{\bfr}_v(\gamma)$ contains at most $2n$ deviations, each bounded by $2\Norm{\bfr}_\infty$ and multiplied by a coefficient in $[0,1]$, giving $\Abs{w^{\bfr}_v(\gamma)} \le 4n\Norm{\bfr}_\infty$.
Differentiating these coefficients bounds the transient contribution by $n^2\Norm{\bfr}_\infty$ and the cycle contribution by $6n^2\Norm{\bfr}_\infty$, yielding the derivative bound.
\end{proof}

\begin{proof}
Let $M = \Norm{\bfr}_\infty$ and $\mu \triangleq \mu^{\text{mp}, \bfr}(v)$.
Since $\gamma \in (\gamma_{\text{bw}}(\bfr), 1)$, we consider a Blackwell optimal pair of positional strategies $(\pi^{\text{bw}}_{\min}, \pi^{\text{bw}}_{\max})$ for $G^{\gamma, \bfr}$.
Observe that the path from vertex $v$ under the strategy profile $(\pi^{\text{bw}}_{\min}, \pi^{\text{bw}}_{\max})$ consists of a finite path of length $t \le n$ and a recurrent cycle of length $c \le n$.
Let $r_k$ denote the reward at step $k$.
Note that the mean-payoff is the cycle's average reward, i.e., $\mu = \frac{1}{c} \sum_{j=0}^{c-1} r_{t+j}$. Therefore, we have $\Abs{\mu} \le M$ and $\sum_{j=0}^{c-1} (r_{t+j} - \mu) = 0$.
We now present the two items of the proof.

\smallskip\noindent\emph{Value decomposition.}
The value vector expands as follows:
\begin{align*}
    \mu^{\gamma, \bfr}(v) &= \sum_{j=0}^{t-1} r_j \gamma^j + \frac{\gamma^t}{1-\gamma^c} \sum_{j=0}^{c-1} r_{t+j} \gamma^j \\
    &= \sum_{j=0}^{t-1} (r_j - \mu) \gamma^j + \frac{\gamma^t}{1-\gamma^c} \sum_{j=0}^{c-1} (r_{t+j} - \mu) \gamma^j + \mu \sum_{j=0}^{t-1} \gamma^j + \frac{\gamma^t \mu}{1-\gamma^c} \sum_{j=0}^{c-1} \gamma^j \\
    &= \sum_{j=0}^{t-1} (r_j - \mu) \gamma^j + \frac{\gamma^t}{1-\gamma^c} \sum_{j=0}^{c-1} (r_{t+j} - \mu) (\gamma^j - 1) + \mu \Bigl( \frac{1-\gamma^t}{1-\gamma} + \frac{\gamma^t}{1-\gamma} \Bigr) \\
    &= \underbrace{ \sum_{j=0}^{t-1} (r_j - \mu) \gamma^j + \sum_{j=1}^{c-1} \biggl(\frac{\sum_{i=0}^{j-1} \gamma^i}{\sum_{i=0}^{c-1} \gamma^i}\biggr) (\mu - r_{t+j}) \gamma^t }_{\triangleq w^{\bfr}_v(\gamma)} + \frac{\mu}{1-\gamma} ,
\end{align*}
where the first equality is due to decomposing the path into a transient path and a cycle, the second equality rearranges the terms to separate the mean-payoff component,
the third equality introduces $-1$ into the cycle sum because $\sum_{j=0}^{c-1} (r_{t+j} - \mu) = 0$, and the final equality factors out $1-\gamma$.
We define $f_j(\gamma) \triangleq \sum_{i=0}^{j-1} \gamma^i$. Then, we rewrite $w^{\bfr}_v(\gamma)$ as:
\begin{align*}
    w^{\bfr}_v(\gamma) &= \sum_{j=0}^{t-1} (r_j - \mu) \gamma^j + \frac{1}{f_c(\gamma)}\sum_{j=1}^{c-1} (\mu - r_{t+j}) \gamma^t f_j(\gamma)\\
    &= \frac{\sum_{j=0}^{t-1} (r_j - \mu) \gamma^j f_c(\gamma) + \sum_{j=1}^{c-1} (\mu - r_{t+j}) \gamma^t f_j(\gamma)}{f_c(\gamma)} .
\end{align*}
Therefore, $w^{\bfr}_v(\gamma)$ is a rational function of $\gamma$, which proves the first item of the statement. We now bound its magnitude and derivative.

\smallskip\noindent
\emph{Magnitude bound.}
For magnitude, we get
\begin{align*}
    \Abs{w^{\bfr}_v(\gamma)} &= \Abs[\bigg]{\sum_{j=0}^{t-1} (r_j - \mu) \gamma^j + \sum_{j=1}^{c-1} (\mu - r_{t+j}) \gamma^t \frac{f_j(\gamma)}{f_c(\gamma)}}\\
    &\le \sum_{j=0}^{t-1} \Abs{r_j - \mu} \gamma^j + \sum_{j=1}^{c-1} \Abs{\mu - r_{t+j}} \gamma^t \frac{f_j(\gamma)}{f_c(\gamma)}\\
    &\le 2M \sum_{j=0}^{t-1} \gamma^j + 2M \sum_{j=1}^{c-1} \gamma^t \frac{f_j(\gamma)}{f_c(\gamma)}\\
    &\le 2M \sum_{j=0}^{t-1} 1 + 2M \sum_{j=1}^{c-1} \frac{f_j(\gamma)}{f_c(\gamma)}\\
    &\le 4Mn = 4n\Norm{\bfr}_\infty ,
\end{align*}
where the first inequality is due to the triangle inequality,
the second inequality is due to the bounds $\Abs{r_j - \mu} \le 2M$ and $\Abs{\mu - r_{t+j}} \le 2M$,
the third inequality is due to $\gamma^j, \gamma^t \le 1$ for all $\gamma \in [0, 1]$, the fourth inequality is due to $f_j(\gamma) \le f_c(\gamma)$ for all $j < c$, and the final inequality is due to $t, c \le n$.

\smallskip\noindent
\emph{Derivative bound.}
For the derivative, we first compute $\Abs[\Big]{\frac{d}{d\gamma} \frac{f_j(\gamma)}{f_c(\gamma)}}$:
\begin{align*}
    \Abs[\Big]{\frac{d}{d\gamma} \frac{f_j(\gamma)}{f_c(\gamma)}} &= \Abs[\bigg]{\frac{(\frac{d}{d\gamma}f_j)(\gamma) f_c(\gamma) - f_j(\gamma) (\frac{d}{d\gamma}f_c)(\gamma)}{f_c(\gamma)^2}}\\
    &= \Abs[\bigg]{\frac{(\frac{d}{d\gamma}f_j)(\gamma)}{f_c(\gamma)} - \frac{f_j(\gamma)}{f_c(\gamma)} \cdot \frac{(\frac{d}{d\gamma}f_c)(\gamma)}{f_c(\gamma)}}\\
    &\le \Abs[\bigg]{\frac{(\frac{d}{d\gamma}f_j)(\gamma)}{f_c(\gamma)}} + \Abs[\bigg]{\frac{(\frac{d}{d\gamma}f_c)(\gamma)}{f_c(\gamma)}}\\
    &\le 2n ,
\end{align*}
where the first equality is due to the quotient rule,
the second equality is due to rearranging the terms,
the first inequality is due to the triangle inequality and the fact that $\Abs{\frac{f_j(\gamma)}{f_c(\gamma)}} \le 1$,
and the second inequality is due to $(\frac{d}{d\gamma}f_j)(\gamma) = \sum_{i=1}^{j-1} i \gamma^{i-1} \le \sum_{i=1}^{c-1} i \gamma^{i-1} = (\frac{d}{d\gamma}f_c)(\gamma) \le c \sum_{i=1}^{c-1} \gamma^{i-1} \le n f_c(\gamma)$. Now, we bound the derivative of $w^{\bfr}_v(\gamma)$ as follows:
\begin{align*}
    \Abs[\Big]{\frac{d}{d\gamma} w^{\bfr}_v(\gamma)}
    &\le \sum_{j=0}^{t-1} \Abs{r_j - \mu} \Abs[\Big]{\frac{d}{d\gamma} \gamma^j} + \sum_{j=1}^{c-1} \Abs{\mu - r_{t+j}} \Abs[\Big]{\frac{d}{d\gamma} \Bigl(\gamma^t \frac{f_j(\gamma)}{f_c(\gamma)}\Bigr)}\\
    &\le 2M \sum_{j=0}^{t-1} \Abs[\Big]{\frac{d}{d\gamma} \gamma^j} + 2M \sum_{j=1}^{c-1} \Abs[\Big]{\frac{d}{d\gamma} \Bigl(\gamma^t \frac{f_j(\gamma)}{f_c(\gamma)}\Bigr)}\\
    &\le 2M \sum_{j=1}^{t-1} j + 2M \sum_{j=1}^{c-1} \Abs[\Big]{\frac{d}{d\gamma} \Bigl(\gamma^t \frac{f_j(\gamma)}{f_c(\gamma)}\Bigr)}\\
    &\le Mn^2 + 2M \sum_{j=1}^{c-1} 3n\\
    &\le 7n^2 M = 7n^2 \Norm{\bfr}_\infty ,
\end{align*}
where the first inequality is due to the triangle inequality,
the second inequality is due to $\Abs{r_j - \mu} \le 2M$ and $\Abs{\mu - r_{t+j}} \le 2M$,
the third inequality is due to $\Abs[\Big]{\frac{d}{d\gamma} \gamma^j} = j \gamma^{j-1} \le j$ for all $\gamma \in [0, 1]$,
the fourth inequality is due to $t \le n$ and $\Abs[\Big]{\frac{d}{d\gamma} \frac{f_j(\gamma)}{f_c(\gamma)}} \le 2n$ as shown above,
and the final inequality is due to $c \le n$.

Hence all the desired items are proved and the result follows.
\end{proof}
\paragraph{Suboptimality gap $\Delta^{\gamma, \bfr}$.}
We define the \emph{suboptimality gap} for the discounted-sum game $G^{\gamma, \bfr}$ as follows. For all edges $(v, v^\prime) \in E$,
\[
\Delta^{\gamma, \bfr} (v, v^\prime) = \bfr (v, v^\prime) + \gamma \mu^{\gamma, \bfr}(v^\prime) - \mu^{\gamma, \bfr} (v) ;
\]
i.e., the quantity $\bfr (v, v^\prime) + \gamma \mu^{\gamma, \bfr}(v^\prime)$ is
the value for switching via the edge, and $\mu^{\gamma, \bfr} (v)$ is the optimal value; and hence the function represents
the relative value for switching via the edge.

\paragraph{Main property.}
For any reward vector $\bfr \in \R ^E$ such that $\mu^{\text{mp}, \bfr}(v) = \mu^{\text{mp}, \bfr}(v^\prime)$, we show that the suboptimality gap is Lipschitz continuous in $\gamma$ over the Blackwell threshold.
In contrast, the value vector $\mu^{\gamma, \bfr} (v)$ may not be Lipschitz continuous in $\gamma$ because the singular term $\frac{\mu^{\text{mp}, \bfr}(v)}{1-\gamma}$ in the value decomposition (\Cref{lemma:value-decomposition}) diverges as $\gamma$ approaches $1$.

\begin{lemma}[Lipschitz continuity]\label{lemma:bounded-improvement-function}
Let $\bfr \in \R^E$ be any reward vector.
For any edge $a = (v, v^\prime) \in E$ where $\mu^{\text{mp}, \bfr}(v) = \mu^{\text{mp}, \bfr}(v^\prime)$, the map $\gamma \mapsto \Delta ^{\gamma, \bfr}(v, v^\prime)$ is $18 n^2 \Norm {\bfr} _{\infty}$-Lipschitz continuous over the interval $(\gamma_{\text{bw}} (\bfr), 1)$.
\end{lemma}

\begin{proof}[Proof overview]
Substituting the value decomposition from \cref{lemma:value-decomposition} into the suboptimality gap, the shared mean-payoff value cancels the singular terms, giving
$\Delta^{\gamma, \bfr}(v, v^\prime) = \bfr(v, v^\prime) - \mu^{\text{mp}, \bfr}(v) + \gamma w^{\bfr}_{v^\prime}(\gamma) - w^{\bfr}_v(\gamma)$.
Therefore, using the magnitude and derivative bounds from \cref{lemma:value-decomposition}, we have
$
    \Abs[\Big]{\frac{d}{d\gamma}\Delta^{\gamma, \bfr}(v, v^\prime)}
    \le 4n\Norm{\bfr}_\infty + 14n^2\Norm{\bfr}_\infty
    \le 18n^2\Norm{\bfr}_\infty
$,
which proves the result.
\end{proof}

\begin{proof}
By \cref{lemma:value-decomposition}, for all $\gamma \in [\gamma_{\text{bw}}(\bfr), 1)$,
the value at vertices $v, v^\prime \in V$ decomposes as:
\[
    \mu^{\gamma, \bfr}(v) = \frac{\mu^{\text{mp}, \bfr}(v)}{1-\gamma} + w^{\bfr}_v(\gamma)\,, \quad \text{and} \quad \mu^{\gamma, \bfr}(v^\prime) = \frac{\mu^{\text{mp}, \bfr}(v^\prime)}{1-\gamma} + w^{\bfr}_{v^\prime}(\gamma)\,,
\]
where $w^{\bfr}_v(\gamma)$ and $w^{\bfr}_{v^\prime}(\gamma)$ are rational functions of $\gamma$ with bounded magnitude and derivative as stated in \cref{lemma:value-decomposition}.

Substituting this decomposition into the suboptimality gap yields:
\begin{align*}
    \Delta^{\gamma, \bfr}(v, v^\prime) &= \bfr(v, v^\prime) + \gamma \Bigl( \frac{\mu^{\text{mp}, \bfr}(v^\prime)}{1-\gamma} + w^{\bfr}_{v^\prime}(\gamma) \Bigr) - \Bigl( \frac{\mu^{\text{mp}, \bfr}(v)}{1-\gamma} + w^{\bfr}_v(\gamma) \Bigr) \\
    &= \bfr(v, v^\prime) - \mu^{\text{mp}, \bfr}(v) + \gamma w^{\bfr}_{v^\prime}(\gamma) - w^{\bfr}_v(\gamma) \,,
\end{align*}
where the first equality is due to substituting the value decomposition, and the second equality is due to the assumption that $\mu^{\text{mp}, \bfr}(v) = \mu^{\text{mp}, \bfr}(v^\prime)$.
Notice that the singular term $\frac{\mu^{\text{mp}, \bfr}(v)}{1-\gamma}$ exactly cancels out.
To bound the Lipschitz constant over $[\gamma_{\text{bw}}(\bfr), 1)$, we evaluate the derivative with respect to $\gamma$:
\[
    \frac{d}{d\gamma} \Delta^{\gamma, \bfr}(v, v^\prime) = w^{\bfr}_{v^\prime}(\gamma) + \gamma (\frac{d}{d\gamma} w^{\bfr}_{v^\prime})(\gamma) - (\frac{d}{d\gamma} w^{\bfr}_v)(\gamma) .
\]
Therefore, we obtain:
\begin{align*}
    \Abs[\Big] {\frac{d}{d\gamma} \Delta^{\gamma, \bfr}(v, v^\prime)}
    &\le \Abs {w^{\bfr}_{v^\prime}(\gamma)} + \Abs {(\frac{d}{d\gamma} w^{\bfr}_{v^\prime})(\gamma)} + \Abs {(\frac{d}{d\gamma} w^{\bfr}_v)(\gamma)}\\
    &\le 4n \Norm{\bfr}_\infty + 7n^2 \Norm{\bfr}_\infty + 7n^2 \Norm{\bfr}_\infty \le 18 n^2 \Norm{\bfr}_\infty \,,
\end{align*}
where the first inequality is due to the triangle inequality and the fact that $\gamma < 1$, and the second inequality is due to the bounds on the magnitude and derivative of $w^{\bfr}_v(\gamma)$ and $w^{\bfr}_{v^\prime}(\gamma)$ as stated in \cref{lemma:value-decomposition}.
Hence, the map $\gamma \mapsto \Delta^{\gamma, \bfr}(v, v^\prime)$ is Lipschitz continuous over $[\gamma_{\text{bw}}(\bfr), 1)$ with Lipschitz constant $18 n^2 \Norm{\bfr}_\infty$.
\end{proof}

\paragraph{Gap at one $\Delta ^{1, \bfr}$.}
The Lipschitz continuity result implies that we can extend the definition of $\Delta^{\gamma, \bfr}$ to $\gamma = 1$ as follows. For any edge $(v, v^\prime) \in E$ where $\mu^{\text{mp}, \bfr}(v) = \mu^{\text{mp}, \bfr}(v^\prime)$, we have
\[
    \Delta^{1, \bfr}(v, v^\prime) \triangleq \lim_{\gamma \uparrow 1} \Delta^{\gamma, \bfr}(v, v^\prime)\,.
\]
Following from \cref{lemma:bounded-improvement-function}, the limit exists and is finite. We also refer to $\Delta^{1, \bfr}$ as the \emph{suboptimality gap} of the mean-payoff game $G^{\text{mp}, \bfr}$.

We conclude by showing that under the smoothed model, with probability one, every vertex has a unique outgoing edge with zero suboptimality gap (\cref{proposition:optimality-uniqueness}), which yields the uniqueness of the Blackwell optimal strategy pair (\cref{corollary:unique-blackwell-strategy}). These results use standard ideas from smoothed analysis.
The following arguments exploit the bounded density of the noise by conditioning on all but one coordinate at a time, which calls for notation that isolates a single edge.

\paragraph{Single-edge decomposition.} To analyze the effect of perturbing a single edge, we decompose the reward and noise vectors along any edge $a \in E$. We denote
\[
    \bfr = (\bfr(a) ; \bfr_{-a}) \quad \text{and} \quad \xi = (\xi_a ; \xi_{-a}) ,
\]
where $\bfr_{-a} \in \R^{E \setminus \{a\}}$ is the vector of rewards for all edges except $a$, defined by $\bfr_{-a}(v, v^\prime) = \bfr(v, v^\prime)$ for all $(v, v^\prime) \in E \setminus \{a\}$, and $\xi_{-a}$ is defined analogously.
For any fixed $\bfr_{-a} \in \R^{E \setminus \{a\}}$, we denote by $\Pr_{\xi_a}[\cdot]$ the probability measure over the single noise coordinate $\xi_a$.

\begin{proposition}[Optimality uniqueness]\label{proposition:optimality-uniqueness}
Under the smoothed model, we have
\[
    \Pr_{\xi}\bigl[
    \forall v\in V,\
    \exists !\, a = (v, v^\prime) \in E
    \text{ such that }
    \mu^{\text{mp}, \bfr}(v) = \mu^{\text{mp}, \bfr}(v^\prime) \text{ and } \Delta^{1, \bfr}(a)=0
    \bigr]=1 .
\]
In words, with probability $1$, for each vertex $v \in V$, there exists a unique outgoing edge $a = (v, v^\prime)$ such that $\mu^{\text{mp}, \bfr}(v) = \mu^{\text{mp}, \bfr}(v^\prime)$ and $\Delta^{1, \bfr}(a) = 0$.
\end{proposition}
\begin{proof}
We separately establish (i) existence, i.e., for every vertex $v \in V$ at least one such edge exists deterministically; and
 (ii) uniqueness, i.e., with probability $1$ no two distinct outgoing edges from $v$ both satisfy the conditions.
 A union bound over the finite vertex set $V$ completes the proof.

\smallskip\noindent\emph{Existence.} Let $\bfr \in \R^E$ be any reward vector and $(\pi^{\text{bw}}_{\min}, \pi^{\text{bw}}_{\max})$ be a Blackwell optimal strategy pair.
Fix a vertex $v \in V$ and let $a = (v, v^\prime) \in E$ be the edge selected at $v$ by this strategy pair.
Note that the path from $v$ under the strategy profile $(\pi^{\text{bw}}_{\min}, \pi^{\text{bw}}_{\max})$ consists of a finite path of length $t \le n$ and a recurrent cycle of length $c \le n$. The analogous path from $v^\prime$ contains the same cycle.
Since the mean-payoff is the average reward of the cycle and Blackwell optimal strategies are also mean-payoff optimal~\citep{bewley1976asymptotic},
we have $\mu^{\text{mp}, \bfr}(v) = \mu^{\text{mp}, \bfr}(v^\prime)$.
Moreover, since $a$ is an optimal action, we have $\Delta^{\gamma, \bfr}(a) = 0$ for all $\gamma \in (\gamma_{\text{bw}}(\bfr), 1)$. By \cref{lemma:bounded-improvement-function}, the limit extends to $\Delta^{1, \bfr}(a) = 0$, which yields the existence.

\smallskip\noindent\emph{Uniqueness.} Fix a vertex $v \in V$ with at least two outgoing edges, otherwise uniqueness holds.
Let $a_1 = (v, v_1)$ and $a_2 = (v, v_2)$ be any two distinct outgoing edges from $v$.
We define the event~$\cE$ that both $a_1$ and $a_2$ satisfy the conditions:
\[
    \cE \triangleq \bigl\{ \bfr \in \R^E : \mu^{\text{mp}, \bfr}(v) = \mu^{\text{mp}, \bfr}(v_1) = \mu^{\text{mp}, \bfr}(v_2), \ \Delta^{1, \bfr}(a_1) = 0, \ \Delta^{1, \bfr}(a_2) = 0 \bigr\} ,
\]
where the mean-payoff equality at $v_2$ is included so that $\Delta^{1, \bfr}(a_2)$ is well-defined. We show that $\Pr_{\xi} [\cE] = 0$, which implies that with probability $1$ at most one outgoing edge from $v$ satisfies the conditions.
Since $a_1 \neq a_2$, a Blackwell optimal strategy must differ from at least one of $a_1$ or $a_2$ at vertex~$v$. Therefore, we decompose $\cE \subseteq \cE_1 \cup \cE_2$, where
\begin{align*}
    \cE_1 &\triangleq \cE \cap \bigl\{ \bfr \in \R^E : \exists\, (\pi^{\text{bw}}_{\min}, \pi^{\text{bw}}_{\max}) \text{ Blackwell optimal for } \bfr \text{ with } \pi^{\text{bw}}(v) \neq v_1 \bigr\} , \\
    \cE_2 &\triangleq \cE \cap \bigl\{ \bfr \in \R^E : \exists\, (\pi^{\text{bw}}_{\min}, \pi^{\text{bw}}_{\max}) \text{ Blackwell optimal for } \bfr \text{ with } \pi^{\text{bw}}(v) \neq v_2 \bigr\} ,
\end{align*}
and $\pi^{\text{bw}}(v) \triangleq \pi^{\text{bw}}_{\min}(v)$ if $v \in V_{\min}$ and $\pi^{\text{bw}}(v) \triangleq \pi^{\text{bw}}_{\max}(v)$ if $v \in V_{\max}$.
We show that $\Pr_{\xi} [\cE_1] = 0$. The bound $\Pr_{\xi} [\cE_2] = 0$ follows by the symmetric argument with the roles of $a_1$ and $a_2$ swapped.
Consider any fixed $\bfr_{-a_1} \in \R ^{E \setminus \{a_1\}}$. We analyze the event $\cE_1$ with respect to the random variable $x \triangleq \bfr(a_1)$. Formally, we prove $\Pr_{\xi_{a_1}} [\cE_1] = 0$.
Let $\mu^{\gamma, \bfr_{-a_1}}$ denote the optimal value vector of the restricted game played on the edge set $E \setminus \{a_1\}$. Since this restricted game is completely independent of $a_1$, its value vector $\mu^{\gamma, \bfr_{-a_1}}$ is a function depending only on $\bfr_{-a_1}$.
Indeed, for any $x \in \R$ such that $\bfr \triangleq (x, \bfr_{-a_1})\in \cE_1$,
there exists a Blackwell optimal strategy pair $(\pi^{\text{bw}}_{\min}, \pi^{\text{bw}}_{\max})$ where $\pi^{\text{bw}}(v) \neq v_1$.
Since the edge $(v, \pi^{\text{bw}}(v)) \in E \setminus \{a_1\}$, removing $a_1$ does not change the optimal value at vertex $v$. Consequently, the value vectors of the full game and the restricted game are equal, i.e., $\mu^{\gamma, \bfr} = \mu^{\gamma, \bfr_{-a_1}}$ for all $\gamma \in (\gamma_{\text{bw}}(\bfr), 1)$.
This equality also implies $\mu^{\text{mp}, \bfr_{-a_1}}(v) = \mu^{\text{mp}, \bfr_{-a_1}}(v_1)$ by \cref{eq:blackwell-optimality}.
Under the event $\cE_1$, taking the limit of suboptimality gap evaluates to:
\[
    \Delta^{1, \bfr}(a_1) = x + \lim_{\gamma \uparrow 1}\ \bigl( \gamma \mu^{\gamma, \bfr_{-a_1}}(v_1) - \mu^{\gamma, \bfr_{-a_1}}(v) \bigr) .
\]
Let $C(\bfr_{-a_1})$ denote the limit term. By \cref{lemma:bounded-improvement-function} applied to the restricted game, $C(\bfr_{-a_1})$ exists and is finite.
Since $C(\bfr_{-a_1})$ is fixed with respect to $x$, the condition $\Delta^{1, \bfr}(a_1) = 0$ requires $x + C(\bfr_{-a_1}) = 0$, or equivalently, $x = -C(\bfr_{-a_1})$.
Let $\lambda(\cE_1)$ denote the Lebesgue measure of the set of $x$ values satisfying this condition. Note that $\lambda(\cE_1) = 0$.
Since the probability density of $x$ is uniformly bounded by $\phi$, we have $\Pr_{\xi_{a_1}} [\cE_1] = 0$.
Taking the expectation over the independent randomness in $\xi_{-a_1}$ gives
\[
    \Pr_{\xi} [\cE_1] = \E_{\xi_{-a_1}} \bigl[ \Pr_{\xi_{a_1}} [\cE_1] \bigr] = 0 .
\]
The symmetric argument applied to $a_2$ gives $\Pr_{\xi} [\cE_2] = 0$, so a union bound yields $\Pr_{\xi} [\cE] \le \Pr_{\xi} [\cE_1] + \Pr_{\xi} [\cE_2] = 0$.
A further union bound over the finitely many pairs of distinct outgoing edges at every vertex $v \in V$ completes the proof.
\end{proof}

\begin{corollary}[Uniqueness of Blackwell optimal strategies]\label{corollary:unique-blackwell-strategy}
Under the smoothed model, we have
\[
    \Pr_{\xi}\bigl[ \exists !\, (\pi^{\text{bw}}_{\min}, \pi^{\text{bw}}_{\max}) \in \Pi_{\min} \times \Pi_{\max} \text{ that is Blackwell optimal for } \bfr \bigr] = 1 .
\]
In words, with probability $1$, there exists a unique Blackwell optimal strategy pair.
\end{corollary}
\begin{proof}
Existence is guaranteed deterministically by the classical Blackwell optimality result~\citep{blackwell1962discrete,bewley1976asymptotic}.
It remains to establish uniqueness with probability $1$.
By \cref{proposition:optimality-uniqueness}, the event
\[
    \cE \triangleq \bigl\{ \bfr \in \R^E : \forall v \in V,\ \exists !\, a^*(v) = (v, v^\prime) \in E \text{ s.t. } \mu^{\text{mp}, \bfr}(v) = \mu^{\text{mp}, \bfr}(v^\prime) \text{ and } \Delta^{1, \bfr}(a^*(v)) = 0 \bigr\}
\]
satisfies $\Pr_{\xi}[\cE] = 1$. We show that on $\cE$, the Blackwell optimal strategy pair is unique.

Fix $\bfr \in \cE$ and let $(\pi^{\text{bw}}_{\min}, \pi^{\text{bw}}_{\max})$ be any Blackwell optimal strategy pair for $\bfr$.
For any vertex $v \in V$, let $a = (v, \pi^{\text{bw}}(v)) \in E$ be the edge selected at $v$ by this pair, where $\pi^{\text{bw}}(v) \triangleq \pi^{\text{bw}}_{\min}(v)$ if $v \in V_{\min}$ and $\pi^{\text{bw}}(v) \triangleq \pi^{\text{bw}}_{\max}(v)$ if $v \in V_{\max}$.
Since Blackwell optimal strategies are also mean-payoff optimal~\citep{bewley1976asymptotic}, we have $\mu^{\text{mp}, \bfr}(v) = \mu^{\text{mp}, \bfr}(\pi^{\text{bw}}(v))$. Moreover, since Blackwell optimal strategies are optimal in discounted-sum games for all $\gamma \in (\gamma_{\text{bw}}(\bfr), 1)$, we have $\Delta^{\gamma, \bfr}(a) = 0$ for such $\gamma$, and by \cref{lemma:bounded-improvement-function}, taking the limit yields $\Delta^{1, \bfr}(a) = 0$.
Since $\bfr \in \cE$, the unique edge satisfying these conditions at $v$ is $a^*(v)$, so $a = a^*(v)$.
Since this holds for every $v \in V$, both $\pi^{\text{bw}}_{\min}$ and $\pi^{\text{bw}}_{\max}$ are uniquely determined, which completes the proof.
\end{proof}

\begin{lemma}[Unique optimal strategies at fixed thresholds]\label{lemma:unique-threshold-optimal-blackwell}
Let $\Gamma \subseteq [0,1)$ be any finite set of discount factors.
Under the smoothed model, with probability $1$, for every $\gamma \in \Gamma$ with $\gamma > \gamma_{\text{bw}}(\bfr)$, the discounted-sum game $G^{\gamma,\bfr}$ has a unique optimal strategy pair, and this pair is the unique Blackwell optimal strategy pair for $\bfr$.
\end{lemma}
\begin{proof}
Fix a discount factor $\gamma \in [0,1)$.
We first show that, with probability $1$, the game $G^{\gamma,\bfr}$ has a unique optimal strategy pair.
It is enough to show that, with probability $1$, no vertex has two distinct outgoing edges that both satisfy the discounted Bellman equality $\Delta^{\gamma,\bfr}(a)=0$.
Indeed, by the Bellman optimality equations, at every vertex at least one outgoing edge satisfies this equality, and if exactly one such edge exists at every vertex, then both players' optimal positional strategies are uniquely determined.

Fix a vertex $v$ with two distinct outgoing edges $a_1=(v,v_1)$ and $a_2=(v,v_2)$, and consider the event
\[
    \cE \triangleq \bigl\{ \bfr \in \R^E : \Delta^{\gamma,\bfr}(a_1)=0 \text{ and } \Delta^{\gamma,\bfr}(a_2)=0 \bigr\}.
\]
We prove that $\Pr_\xi[\cE]=0$.
Condition on all rewards except $a_1$, writing $x=\bfr(a_1)$ and $\bfr=(x;\bfr_{-a_1})$.
On the event $\cE$, the edge $a_2$ is also optimal at $v$; hence there exists an optimal strategy pair for $G^{\gamma,\bfr}$ whose strategy at $v$ selects $a_2$ and therefore does not use $a_1$.
Removing $a_1$ from the game does not change the value vector at discount $\gamma$: if $v \in V_{\min}$, Min still has an optimal strategy that avoids $a_1$, while removing a Min edge can only increase the value; if $v \in V_{\max}$, Max still has an optimal strategy that avoids $a_1$, while removing a Max edge can only decrease the value.
Thus $\mu^{\gamma,\bfr}=\mu^{\gamma,\bfr_{-a_1}}$ on $\cE$, where $\mu^{\gamma,\bfr_{-a_1}}$ denotes the value vector of the game with edge $a_1$ removed.
Consequently, on $\cE$,
\[
    0=\Delta^{\gamma,\bfr}(a_1)
    = x + \gamma \mu^{\gamma,\bfr_{-a_1}}(v_1) - \mu^{\gamma,\bfr_{-a_1}}(v).
\]
For the fixed conditioned vector $\bfr_{-a_1}$, the last two terms are independent of $x$, so this equality can hold for at most one value of $x$.
Since the conditional density of $x=\bfr(a_1)$ is bounded, the conditional probability of $\cE$ is zero.
Integrating over $\bfr_{-a_1}$ gives $\Pr_\xi[\cE]=0$.
A union bound over all vertices and all pairs of distinct outgoing edges shows that $G^{\gamma,\bfr}$ has a unique optimal strategy pair with probability $1$.

Since $\Gamma$ is finite, another union bound implies that, with probability $1$, every game $G^{\gamma,\bfr}$ with $\gamma \in \Gamma$ has a unique optimal strategy pair.
Intersect this event with the probability-one event of \cref{corollary:unique-blackwell-strategy}.
On the intersection, if $\gamma \in \Gamma$ and $\gamma>\gamma_{\text{bw}}(\bfr)$, then the unique Blackwell optimal pair is optimal for $G^{\gamma,\bfr}$ by the definition of the Blackwell threshold.
Since the discounted optimal pair at $\gamma$ is unique, it must coincide with the unique Blackwell optimal pair.
\end{proof}

\subsection{Ingredient~2: Suboptimality Gap Separation}\label{sec:ingredient-2}
We now establish a probabilistic separation result for the suboptimality gap under the smoothed model (Ingredient~2 from \cref{sec:techniques-overview}).
For any fixed edge $a = (v, v^\prime)$, the event that (i)~vertices $v$ and $v^\prime$ have the same mean-payoff value and (ii)~the suboptimality gap at $\gamma=1$ is small and non-zero, has low probability.

\begin{lemma}[Suboptimality separation]\label{lemma:value-separation-at-one}
Let $G$ be any game graph and let $a = (v, v^\prime) \in E$ be any fixed edge.
Under the smoothed model, for any $\delta > 0$, we have
    \[
    \Pr _{\xi} \bigl[ \mu^{\text{mp}, \bfr}(v) = \mu^{\text{mp}, \bfr}(v^\prime) \text{ and } 0 < \Abs[\big]{\Delta^{1, \bfr}(a)} < \delta \bigr] \le 2\delta \phi .
    \]
\end{lemma}

\begin{proof}[Proof overview]
Consider the event $\cE \triangleq \bigl\{ \bfr \in \R^E : \mu^{\text{mp}, \bfr}(v) = \mu^{\text{mp}, \bfr}(v^\prime) \text{ and } 0 < \Abs[\big]{\Delta^{1, \bfr}(a)} < \delta \bigr\}$.
If $a$ is the only outgoing edge from $v$, then $\Delta^{1, \bfr}(a) = 0$ whenever well-defined, so $\cE$ is empty.
Otherwise, fix the rewards $\bfr_{-a}$ on all edges except $a$ and let $x \triangleq \bfr(a)$.
For any $x$ such that $(x;\bfr_{-a}) \in \cE$, the edge $a$ is strictly suboptimal for all $\gamma$ sufficiently close to $1$.
Deleting $a$ therefore leaves the value vector unchanged and preserves the shared mean-payoff value of $v$ and $v^\prime$.
Thus,
\[
    \Delta^{1, \bfr}(a)
    = x + \lim_{\gamma \uparrow 1} \bigl( \gamma \mu^{\gamma, \bfr_{-a}}(v^\prime) - \mu^{\gamma, \bfr_{-a}}(v) \bigr) .
\]
By \cref{lemma:value-decomposition}, the limit term is finite and independent of $x$.
Hence the set of $x$ values producing $\cE$ has length at most $2\delta$, and therefore conditional probability at most $2\delta\phi$.
Since this bound holds for every fixed $\bfr_{-a}$, the unconditional bound follows.
\end{proof}

\begin{proof}
If $a$ is the only outgoing edge from $v$, it must be selected by any valid strategy. Thus, $\Delta^{1, \bfr}(a) = 0$ whenever well-defined, and the event has probability $0$. Assume $a$ is not the only outgoing edge from $v$.
Consider any fixed $\bfr_{-a} \in \R ^{E \setminus \{a\}}$.
We bound the probability of the event
\[
    \cE \triangleq \bigl\{ \bfr(a) \in \R : \mu^{\text{mp}, \bfr}(v) = \mu^{\text{mp}, \bfr}(v^\prime) \text{ and } 0 < \Abs[\big]{\Delta^{1, \bfr}(a)} < \delta \bigr\} .
\]
Formally, we prove $\Pr_{\xi_a} [\cE] \le 2\delta \phi$.

Let $\mu^{\gamma, \bfr_{-a}}$ denote the optimal value vector of the restricted game played on the edge set~$E \setminus \{a\}$. Since this restricted game is completely independent of edge $a$, its value vector $\mu^{\gamma, \bfr_{-a}}$ is a function depending only on $\bfr_{-a}$.
For any $x \in \cE$, let $\bfr = (x, \bfr_{-a})$. Then, we have $\Delta^{1, \bfr}(a) \neq 0$. This implies that the edge $a$ is strictly suboptimal under the Blackwell optimal strategies. Consequently, for $\gamma$ sufficiently close to $1$, the optimal strategy does not select $a$, meaning the value vectors of the full game and the restricted game exactly match: $\mu^{\gamma, \bfr} = \mu^{\gamma, \bfr_{-a}}$. Taking the limit $\gamma \uparrow 1$ of $(1-\gamma) \mu^{\gamma, \bfr}$ and applying \cref{eq:blackwell-optimality} gives $\mu^{\text{mp}, \bfr_{-a}}(u) = \mu^{\text{mp}, \bfr}(u)$ for all $u \in V$. In particular, the assumption $\mu^{\text{mp}, \bfr}(v) = \mu^{\text{mp}, \bfr}(v^\prime)$ implies $\mu^{\text{mp}, \bfr_{-a}}(v) = \mu^{\text{mp}, \bfr_{-a}}(v^\prime)$.

Under this condition, the limit of the suboptimality gap evaluates to:
\[
    \Delta^{1, \bfr}(a) = x + \lim_{\gamma \uparrow 1}\ \bigl( \gamma \mu^{\gamma, \bfr_{-a}}(v^\prime) - \mu^{\gamma, \bfr_{-a}}(v) \bigr) .
\]
Let $C(\bfr_{-a})$ denote the limit term. Applying the value decomposition (\cref{lemma:value-decomposition}) to the restricted game and using $\mu^{\text{mp}, \bfr_{-a}}(v) = \mu^{\text{mp}, \bfr_{-a}}(v^\prime)$, the singular $\frac{1}{1-\gamma}$ terms cancel, so $C(\bfr_{-a})$ exists and is finite.
Since $C(\bfr_{-a})$ is fixed with respect to $x$, the condition $0 < \Abs[\big]{x + C(\bfr_{-a})} < \delta$ requires the value of $x$ to fall strictly within the interval $\bigl(-C(\bfr_{-a}) - \delta, -C(\bfr_{-a}) + \delta \bigr)$.
Let $\lambda(\cE)$ denote the Lebesgue measure of the set of $x$ values satisfying this condition. Note that $\lambda(\cE) \le 2\delta$.
Since the probability density of $x$ is uniformly bounded by $\phi$, the probability of the event satisfies $\Pr_{\xi_a} [\cE] \le 2\delta \phi$.
Taking expectation over the independent randomness in $\xi_{-a}$ gives
\[
    \Pr_{\xi} [\cE] = \E_{\xi_{-a}} \bigl[ \Pr_{\xi_a} [\cE] \bigr] \le 2\delta \phi ,
\]
which completes the proof.
\end{proof}

\subsection{Ingredient~3: Mean-payoff Value Separation}\label{sec:ingredient-3}
We establish Ingredient~3 from \cref{sec:techniques-overview}: under the smoothed model, the mean-payoff values of any two vertices are either the same or separated by an inverse-polynomial gap with high probability.

\begin{lemma}[Mean-payoff separation]\label{lemma:non-ergodic-value-separation}
Let $G$ be any game graph and let $v, v^\prime \in V$. Under the smoothed model, for any $B > 0$, we have
\[
    \Pr_{\xi} \Bigl[ 0 < \Abs[\big]{ \mu^{\text{mp}, \bfr}(v) - \mu^{\text{mp}, \bfr}(v^\prime) } < B \Bigr] \le 2 B m n \phi .
\]
\end{lemma}

\begin{proof}[Proof overview]
With probability $1$ the Blackwell optimal pair is unique; fix it and let $C^{\bfr}_v$ be the recurrent cycle reached from $v$, of average reward $\mu^{\text{mp}, \bfr}(v)$.
If $\mu^{\text{mp}, \bfr}(v) > \mu^{\text{mp}, \bfr}(v^\prime)$, then no edge of $C^{\bfr}_v$ is reachable from $v^\prime$ under the optimal Min strategy $\pi^{\text{bw}}_{\min}$, since otherwise Max could reach and repeat $C^{\bfr}_v$ from $v^\prime$ and secure $\mu^{\text{mp}, \bfr}(v) > \mu^{\text{mp}, \bfr}(v^\prime)$ against $\pi^{\text{bw}}_{\min}$.
Assume $0 < \mu^{\text{mp}, \bfr}(v) - \mu^{\text{mp}, \bfr}(v^\prime) < B$, fix an edge $a \in E(C^{\bfr}_v)$, condition on $\bfr_{-a}$, and let $x \triangleq \bfr(a)$.
Since $a$ is unreachable from $v^\prime$, the value $\mu^{\text{mp}, (x; \bfr_{-a})}(v^\prime)$ is a constant $\alpha$ independent of $x$, whereas $x \mapsto \mu^{\text{mp}, (x; \bfr_{-a})}(v)$ is nondecreasing with derivative at least $1/n$ almost everywhere on the set where $a$ lies on the cycle reached from $v$.
This set therefore maps into $(\alpha, \alpha + B)$ and has Lebesgue measure at most $nB$; multiplying by the density bound $\phi$, summing over the $m$ edges, and adding the symmetric case gives $2 m n B \phi$.
\end{proof}

To prove \cref{lemma:non-ergodic-value-separation}, we first define some useful notation. We then establish a reachability separation lemma, which is a deterministic statement about the structure of Blackwell optimal strategies. Finally, we use this lemma to prove the mean-payoff value separation result.

\paragraph{Cycles.} Let $\cC$ denote the set of simple directed cycles of $G$. For a cycle $C \in \cC$, let $E(C)$ and $V(C)$ be its edge set and vertex set, and let
\[
    \avg_{\bfr}(C) \triangleq \frac{1}{\Abs{E(C)}} \sum_{a \in E(C)} \bfr(a) .
\]

\paragraph{Reachable vertices and edges.}
For a strategy $\pi_{\min} \in \Pi_{\min}$ of the Min Player, let $V_{\pi_{\min}}(v)$ be the set of vertices reachable from $v$ after fixing Min to $\pi_{\min}$ and leaving all Max edges available.
Let $E_{\pi_{\min}}(v)$ be the set of edges reachable from $v$ in this one-player game, i.e.,
\[
    E_{\pi_{\min}}(v) \triangleq \bigl\{ (v^\prime, v^{\prime\prime}) \in E : v^\prime \in V_{\pi_{\min}}(v) \cap V_{\max} \bigr\} \cup \bigl\{ (v^\prime, \pi_{\min}(v^\prime)) : v^\prime \in V_{\pi_{\min}}(v) \cap V_{\min} \bigr\} .
\]
Symmetrically, for a strategy $\pi_{\max} \in \Pi_{\max}$ of the Max Player, let $V_{\pi_{\max}}(v)$ be the set of vertices reachable from $v$ after fixing Max to $\pi_{\max}$ and leaving all Min edges available, and let
\[
    E_{\pi_{\max}}(v) \triangleq \bigl\{ (v^\prime, v^{\prime\prime}) \in E : v^\prime \in V_{\pi_{\max}}(v) \cap V_{\min} \bigr\} \cup \bigl\{ (v^\prime, \pi_{\max}(v^\prime)) : v^\prime \in V_{\pi_{\max}}(v) \cap V_{\max} \bigr\} .
\]
For a strategy $\pi_{\min} \in \Pi_{\min}$, the cycles $C \in \cC$ with $E(C) \subseteq E_{\pi_{\min}}(v)$ are exactly the simple directed cycles reachable from $v$ after fixing $\pi_{\min}$, and likewise for $\pi_{\max} \in \Pi_{\max}$.

\begin{lemma}[Reachability separation]\label{lemma:reachability-separation}
Let $\bfr \in \R^E$, let $(\pi^{\text{bw}}_{\min}, \pi^{\text{bw}}_{\max})$ be a Blackwell optimal strategy pair, and let $v, v^\prime \in V$.
The cycle $C^{\bfr}_v$ is the recurrent cycle reached from $v$ under $(\pi^{\text{bw}}_{\min}, \pi^{\text{bw}}_{\max})$. 
\begin{enumerate}
    \item If $\mu^{\text{mp}, \bfr}(v) > \mu^{\text{mp}, \bfr}(v^\prime)$, then
    \[
        E(C^{\bfr}_v) \cap E_{\pi^{\text{bw}}_{\min}}(v^\prime) = \emptyset ,
    \]
    \item If $\mu^{\text{mp}, \bfr}(v) < \mu^{\text{mp}, \bfr}(v^\prime)$, then
    \[
        E(C^{\bfr}_v) \cap E_{\pi^{\text{bw}}_{\max}}(v^\prime) = \emptyset .
    \]
\end{enumerate}
\end{lemma}

\begin{proof}
Since the pair $(\pi^{\text{bw}}_{\min}, \pi^{\text{bw}}_{\max})$ is Blackwell optimal, it is also mean-payoff optimal~\citep{bewley1976asymptotic}. Therefore, $\pi^{\text{bw}}_{\min}$ is an optimal min strategy and $\pi^{\text{bw}}_{\max}$ an optimal max strategy for the game $G^{\text{mp}, \bfr}$.
Suppose $E(C^{\bfr}_v) \cap E_{\pi^{\text{bw}}_{\min}}(v^\prime) \neq \emptyset$. Then, there exists a vertex $u \in C^{\bfr}_v \cap V_{\pi^{\text{bw}}_{\min}}(v^\prime)$. Since $C^{\bfr}_v$ follows $\pi^{\text{bw}}_{\min}$ at every Min vertex, traversing it from $u$ gives $E(C^{\bfr}_v) \subseteq E_{\pi^{\text{bw}}_{\min}}(v^\prime)$, so some Max strategy $\pi_{\max}$ makes $\rho(v^\prime, \pi^{\text{bw}}_{\min}, \pi_{\max})$ reach $C^{\bfr}_v$ and cycle it forever, giving $u^{\text{mp}, \bfr}(\rho(v^\prime, \pi^{\text{bw}}_{\min}, \pi_{\max})) = \avg_{\bfr}(C^{\bfr}_v)$. As $\pi^{\text{bw}}_{\min}$ is an optimal min strategy, this payoff is at most $\mu^{\text{mp}, \bfr}(v^\prime)$, and therefore
\[
    \mu^{\text{mp}, \bfr}(v) = \avg_{\bfr}(C^{\bfr}_v) \le \mu^{\text{mp}, \bfr}(v^\prime) ,
\]
contradicting $\mu^{\text{mp}, \bfr}(v) > \mu^{\text{mp}, \bfr}(v^\prime)$.
Hence $E(C^{\bfr}_v) \cap E_{\pi^{\text{bw}}_{\min}}(v^\prime) = \emptyset$.

The proof of the second statement is symmetric.
Suppose $E(C^{\bfr}_v) \cap E_{\pi^{\text{bw}}_{\max}}(v^\prime) \neq \emptyset$. Then, there exists a vertex $u \in C^{\bfr}_v \cap V_{\pi^{\text{bw}}_{\max}}(v^\prime)$. Since $C^{\bfr}_v$ follows $\pi^{\text{bw}}_{\max}$ at every Max vertex, traversing it from $u$ gives $E(C^{\bfr}_v) \subseteq E_{\pi^{\text{bw}}_{\max}}(v^\prime)$, so some Min strategy $\pi_{\min}$ makes $\rho(v^\prime, \pi_{\min}, \pi^{\text{bw}}_{\max})$ reach $C^{\bfr}_v$ and cycle it forever, giving $u^{\text{mp}, \bfr}(\rho(v^\prime, \pi_{\min}, \pi^{\text{bw}}_{\max})) = \avg_{\bfr}(C^{\bfr}_v)$. As $\pi^{\text{bw}}_{\max}$ is an optimal max strategy, this payoff is at least $\mu^{\text{mp}, \bfr}(v^\prime)$, and therefore
\[
    \mu^{\text{mp}, \bfr}(v) = \avg_{\bfr}(C^{\bfr}_v) \ge \mu^{\text{mp}, \bfr}(v^\prime) ,
\]
contradicting $\mu^{\text{mp}, \bfr}(v) < \mu^{\text{mp}, \bfr}(v^\prime)$.
Hence $E(C^{\bfr}_v) \cap E_{\pi^{\text{bw}}_{\max}}(v^\prime) = \emptyset$.
\end{proof}

\begin{proof}[Proof of \cref{lemma:non-ergodic-value-separation}]
If $v = v^\prime$, the event is empty; assume $v \neq v^\prime$.
Let
\[
    \cE_{\text{bw}} \triangleq \bigl\{ \bfr \in \R^E : \exists !\, (\pi^{\text{bw}}_{\min}, \pi^{\text{bw}}_{\max}) \in \Pi_{\min} \times \Pi_{\max} \text{ that is Blackwell optimal for } \bfr \bigr\} .
\]
By \cref{corollary:unique-blackwell-strategy}, $\Pr_\xi[\cE_{\text{bw}}] = 1$, so it suffices to bound the two bad events
\begin{align*}
    \cB^+ &\triangleq \cE_{\text{bw}} \cap \bigl\{ \bfr : 0 < \mu^{\text{mp}, \bfr}(v) - \mu^{\text{mp}, \bfr}(v^\prime) < B \bigr\} , \\
    \cB^- &\triangleq \cE_{\text{bw}} \cap \bigl\{ \bfr : 0 < \mu^{\text{mp}, \bfr}(v^\prime) - \mu^{\text{mp}, \bfr}(v) < B \bigr\} .
\end{align*}
We bound $\Pr_\xi[\cB^+]$; the bound on $\Pr_\xi[\cB^-]$ is identical after exchanging $v$ and $v^\prime$.

Given an edge $a \in E$ and rewards $\bfr_{-a}$, when $(x; \bfr_{-a}) \in \cE_{\text{bw}}$ let $\pi^{\text{bw}}_{\min}$ be its Blackwell optimal Min strategy and $C^{(x; \bfr_{-a})}_v$ the recurrent cycle reached from $v$, and define the slice
\[
    S_a^+(\bfr_{-a}) \triangleq \bigl\{ x \in \R : (x; \bfr_{-a}) \in \cE_{\text{bw}},\ 0 < \mu^{\text{mp}, (x; \bfr_{-a})}(v) - \mu^{\text{mp}, (x; \bfr_{-a})}(v^\prime) < B,\ a \in E(C^{(x; \bfr_{-a})}_v),\ a \notin E_{\pi^{\text{bw}}_{\min}}(v^\prime) \bigr\} .
\]
For rewards $\bfr \in \cB^+$, let $C^{\bfr}_v$ be the recurrent cycle reached from $v$ under its unique Blackwell optimal pair. Therefore, we have $\avg_{\bfr}(C^{\bfr}_v) = \mu^{\text{mp}, \bfr}(v)$.
By \cref{lemma:reachability-separation}, for all edges $a \in E(C^{\bfr}_v)$, we have $a \notin E_{\pi^{\text{bw}}_{\min}}(v^\prime)$. Thus, $\bfr(a) \in S_a^+(\bfr_{-a})$.
Consequently, we have
\[
    \cB^+ \subseteq \bigcup_{a \in E} \bigl\{ \bfr : \bfr(a) \in S_a^+(\bfr_{-a}) \bigr\} .
\]
Moreover, we show that $\mu^{\text{mp}, (x; \bfr_{-a})}(v^\prime)$ is constant on $S_a^+(\bfr_{-a})$.
For a fixed Min strategy $\pi_{\min}$, every play from $v^\prime$ is a lasso whose recurrent cycle $C$ satisfies $E(C) \subseteq E_{\pi_{\min}}(v^\prime)$, and the Max Player can reach and repeat any such cycle; hence the largest payoff the Max Player secures from $v^\prime$ against $\pi_{\min}$ is $\max\{ \avg_{(x; \bfr_{-a})}(C) : C \in \cC,\ E(C) \subseteq E_{\pi_{\min}}(v^\prime) \}$. Since the mean-payoff value is the min-max of the payoff,
\[
    \mu^{\text{mp}, (x; \bfr_{-a})}(v^\prime)
    = \min_{\pi_{\min} \in \Pi_{\min}} \; \max_{\substack{C \in \cC \\ E(C) \subseteq E_{\pi_{\min}}(v^\prime)}} \avg_{(x; \bfr_{-a})}(C) .
\]
For $x \in S_a^+(\bfr_{-a})$, the Blackwell optimal $\pi^{\text{bw}}_{\min}$ is an optimal min strategy, so it attains this minimum, and it satisfies $a \notin E_{\pi^{\text{bw}}_{\min}}(v^\prime)$; hence the minimum is attained among the strategies with $a \notin E_{\pi_{\min}}(v^\prime)$:
\[
    \mu^{\text{mp}, (x; \bfr_{-a})}(v^\prime)
    = \min_{\substack{\pi_{\min} \in \Pi_{\min} \\ a \notin E_{\pi_{\min}}(v^\prime)}} \; \max_{\substack{C \in \cC \\ E(C) \subseteq E_{\pi_{\min}}(v^\prime)}} \avg_{(x; \bfr_{-a})}(C) .
\]
Every cycle on the right-hand side avoids $a$, so its average does not involve $x = \bfr(a)$; hence $\mu^{\text{mp}, (x; \bfr_{-a})}(v^\prime)$ is a constant $\alpha$ depending only on $\bfr_{-a}$ for all $x \in S_a^+(\bfr_{-a})$.
Let $f(x) \triangleq \mu^{\text{mp}, (x; \bfr_{-a})}(v)$. The function $f$ is continuous, piecewise affine, nondecreasing, and absolutely continuous because it is the min-max over finitely many positional strategy pairs, and for each such pair the recurrent-cycle mean is an affine function of $x$ with slope $0$ or $1/c$ for some $c \le n$.
Let
\[
    D \triangleq \bigcup_{\substack{C, C^\prime \in \cC \\ \avg_{(\cdot; \bfr_{-a})}(C) \,\not\equiv\, \avg_{(\cdot; \bfr_{-a})}(C^\prime)}} \bigl\{ x \in \R : \avg_{(x; \bfr_{-a})}(C) = \avg_{(x; \bfr_{-a})}(C^\prime) \bigr\} ,
\]
which is finite and, since each affine piece of $f$ coincides with $\avg_{(x; \bfr_{-a})}(C)$ for some $C \in \cC$, contains every breakpoint of $f$.
Fix a connected component $J$ of $\R \setminus D$ with $J \cap S_a^+(\bfr_{-a}) \neq \emptyset$.
Because every breakpoint of $f$ belongs to $D$, the restriction of $f$ to $J$ consists of a single affine piece.
Hence there exists a cycle $\widehat C \in \cC$ such that
\[
    f(x)=\avg_{(x;\bfr_{-a})}(\widehat C)
    \qquad\text{for every }x\in J.
\]
Choose $x_0\in J\cap S_a^+(\bfr_{-a})$ and let
\[
    C_0\triangleq C^{(x_0;\bfr_{-a})}_v .
\]
By the definition of $S_a^+(\bfr_{-a})$, we have $(x_0;\bfr_{-a})\in\cE_{\text{bw}}$ and $a\in E(C_0)$.
The Blackwell optimal strategy pair at $(x_0;\bfr_{-a})$ induces the recurrent cycle $C_0$ from $v$, and therefore
\[
    f(x_0)
    =\avg_{(x_0;\bfr_{-a})}(C_0)
    =\avg_{(x_0;\bfr_{-a})}(\widehat C).
\]
If the two affine functions
$\avg_{(\cdot;\bfr_{-a})}(C_0)$ and
$\avg_{(\cdot;\bfr_{-a})}(\widehat C)$ were not identical, their equality at $x_0$ would imply $x_0\in D$ by the definition of $D$, contradicting $x_0\in J\subseteq\R\setminus D$.
Consequently, these affine functions are identical, and hence
\[
    f(x)=\avg_{(x;\bfr_{-a})}(C_0)
    \qquad\text{for every }x\in J.
\]
The recurrent cycle $C_0$ is simple, contains $a$, and has at most $n$ edges.
Thus, the variable reward $x=\bfr(a)$ occurs exactly once in the sum defining its cycle average, and therefore
\[
    f^\prime(x) = \frac{1}{\Abs{E(C_0)}} \ge \frac{1}{n} \qquad \text{for all } x \in J .
\]
Every $x \in S_a^+(\bfr_{-a}) \setminus D$ lies in such a component, so $f^\prime(x) \ge 1/n$ for almost every $x \in S_a^+(\bfr_{-a})$.
Since $\mu^{\text{mp}, (x; \bfr_{-a})}(v^\prime) = \alpha$ on $S_a^+(\bfr_{-a})$, we have $S_a^+(\bfr_{-a}) \subseteq f^{-1}\bigl((\alpha, \alpha + B)\bigr)$. Therefore, since $f$ is nondecreasing and absolutely continuous, we have
\[
    \lambda\bigl(S_a^+(\bfr_{-a})\bigr)
    \le n \int_{S_a^+(\bfr_{-a})} f^\prime(x)\,dx
    \le n \int_{f^{-1}((\alpha, \alpha + B))} f^\prime(x)\,dx
    \le nB ,
\]
where $\lambda(\cdot)$ denotes the Lebesgue measure.

Conditioning on $\bfr_{-a}$, the reward $\bfr(a) = \bfr_0(a) + \xi_a$ has density at most $\phi$ and is independent of $\bfr_{-a}$, so
\[
    \Pr_\xi[\cB^+]
    \le \sum_{a \in E} \Pr_\xi\bigl[ \bfr(a) \in S_a^+(\bfr_{-a}) \bigr]
    \le \sum_{a \in E} \phi \, \E_{\bfr_{-a}}\bigl[ \lambda\bigl(S_a^+(\bfr_{-a})\bigr) \bigr]
    \le mnB\phi .
\]
The identical argument gives $\Pr_\xi[\cB^-] \le mnB\phi$, and therefore
\[
    \Pr_\xi\Bigl[ 0 < \Abs[\big]{\mu^{\text{mp}, \bfr}(v) - \mu^{\text{mp}, \bfr}(v^\prime)} < B \Bigr]
    \le \Pr_\xi[\cB^+] + \Pr_\xi[\cB^-]
    \le 2mnB\phi .
\]
\end{proof}

\section{Discounted-sum games}\label{sec:ergodic-games}
In this section, we focus our attention on discounted-sum games under the smoothed model.
To establish smoothed polynomial runtime, we first derive a tail bound on the Blackwell threshold $\gamma_{\text{bw}}(\bfr)$ by combining the mathematical ingredients from \cref{sec:math-properties} (\cref{sec:tail-bound-general}).
We then introduce a restarting variant of policy iteration and establish a total iteration bound for it in terms of $\gamma_{\text{bw}}(\bfr)$ (\cref{sec:repi-deterministic-bound}).
Combining this bound with the tail bound yields smoothed polynomial runtime with high probability (\cref{sec:threshold-PI-and-high-prob-bound-ergodic}).
Combining it with a bound on the expected truncated inverse gap yields smoothed polynomial expected runtime (\cref{sec:expected-running-time-ergodic}).

\subsection{Tail bound}\label{sec:tail-bound-general}
In this subsection, we establish a tail bound on the Blackwell threshold $\gamma_{\text{bw}}(\bfr)$ under the smoothed model.
The resulting bound shows that $\gamma_{\text{bw}}(\bfr)$ is bounded away from $1$ with high probability and serves as the foundation for both the high probability and expected runtime analyses in the sequel.
We first define the variable $K$ which is a polynomial in the input parameters $m$, $n$, and $\phi$:
\begin{equation}
\label{eq:K-definition}
    K \triangleq 92 m^2 n^2 (\phi + 1).
\end{equation}
\begin{lemma}[Tail bound]\label{lemma:general-bw-threshold-tail}
Let $G$ be any game graph.
Under the smoothed model, for any $x \ge 1$, with the bounding function defined as $M(x) \triangleq 1 + \frac{1}{\theta} \ln (m x)$, we have
\[
    \Pr_{\xi} \bigl[ \gamma_{\text{bw}}(\bfr) > 1 - \tfrac{1}{x} \bigr] \le \frac{K M(x)}{x} .
\]
\end{lemma}
\begin{proof}[Proof overview]
We combine the reward tail bound~\cref{eq:reward-tail-bound} with the suboptimality gap and mean-payoff-value separation events from \cref{lemma:value-separation-at-one} and \cref{lemma:non-ergodic-value-separation}.
Condition on the intersection of these events, \cref{lemma:value-decomposition} and \cref{lemma:bounded-improvement-function} imply that every non-optimal edge retains its sign for all $\gamma > 1-1/x$, and hence $\gamma_{\text{bw}}(\bfr) \le 1-1/x$.
Choosing the separation parameters proportional to $M(x)/x$ and applying union bounds over the edges shows that the complement of this event has probability at most $KM(x)/x$.
\end{proof}

\begin{proof}
For any $M \ge 1$, $B > 0$, and $\delta > 0$, define the events
\begin{align*}
    \cE_1(M) &\triangleq \bigl\{ \bfr \in \R^E : \Norm{\bfr}_\infty \le M \bigr\} , \\
    \cE_2(B) &\triangleq \bigl\{ \bfr \in \R^E : \Abs[\big]{\mu^{\text{mp}, \bfr}(v) - \mu^{\text{mp}, \bfr}(v^\prime)} \in \{0\} \cup [B, \infty) \text{ for all } (v, v^\prime) \in E \bigr\} , \\
    \cE_3(\delta) &\triangleq \bigl\{ \bfr \in \R^E : \Abs[\big]{\Delta^{1, \bfr}(a)} \in \{0\} \cup [\delta, \infty) \text{ for all } a = (v, v^\prime) \in E \text{ with } \mu^{\text{mp}, \bfr}(v) = \mu^{\text{mp}, \bfr}(v^\prime) \bigr\} , \\
    \cE_4 &\triangleq \bigl\{ \bfr \in \R^E : \forall v \in V,\ \exists !\, a = (v, v^\prime) \in E \text{ such that } \mu^{\text{mp}, \bfr}(v) = \mu^{\text{mp}, \bfr}(v^\prime) \text{ and } \Delta^{1, \bfr}(a) = 0 \bigr\} .
\end{align*}
Let $L \triangleq 18 n^2 M$. Conditioned on $\cE_1(M) \cap \cE_2(B) \cap \cE_3(\delta) \cap \cE_4$, we show that $\gamma_{\text{bw}}(\bfr) \le \max\bigl(0, 1 - \frac{\delta}{2L}, 1 - \frac{B}{10 n M}\bigr) = \max\bigl(0, 1 - \frac{\delta}{36 n^2 M}, 1 - \frac{B}{10 n M}\bigr)$.

Indeed, consider $\bfr \in \cE_1(M) \cap \cE_2(B) \cap \cE_3(\delta) \cap \cE_4$. Since $\bfr \in \cE_4$, by \cref{proposition:optimality-uniqueness}, for each vertex $v \in V$ there exists a unique edge $a^\star(v) = (v, v^\prime) \in E$ with $\mu^{\text{mp}, \bfr}(v) = \mu^{\text{mp}, \bfr}(v^\prime)$ and $\Delta^{1, \bfr}(a^\star(v)) = 0$. Therefore, by \cref{corollary:unique-blackwell-strategy}, there exists a unique Blackwell optimal strategy pair $(\pi^{\text{bw}}_{\min}, \pi^{\text{bw}}_{\max})$, which selects $a^\star(v)$ at each vertex $v \in V$.

By \cref{lemma:bounded-improvement-function} and \cref{lemma:value-decomposition}, since $\bfr \in \cE_1(M)$, the $L$-Lipschitz bound on the suboptimality gap and the value decomposition both hold on $(\gamma_{\text{bw}}(\bfr), 1)$. We argue by contradiction: suppose $\gamma_{\text{bw}}(\bfr) > \max\bigl(1 - \frac{\delta}{2L}, 1 - \frac{B}{10nM}\bigr)$. Then for any $\gamma \in (\gamma_{\text{bw}}(\bfr), 1)$, $1 - \gamma < 1 - \gamma_{\text{bw}}(\bfr) < \min\bigl(\frac{\delta}{2L}, \frac{B}{10nM}\bigr)$. We bound $\Abs[\big]{\Delta^{\gamma, \bfr}(a)}$ for any non-optimal edge $a = (v, v^\prime) \neq a^\star(v)$ on $(\gamma_{\text{bw}}(\bfr), 1)$, splitting into two cases.

\smallskip\noindent
\textbf{Case $\mu^{\text{mp}, \bfr}(v) = \mu^{\text{mp}, \bfr}(v^\prime)$.}
Since $\bfr \in \cE_3(\delta)$ and $a$ is non-optimal, $\Abs[\big]{\Delta^{1, \bfr}(a)} \ge \delta$. Then:
\begin{align*}
    \Abs[\big]{\Delta^{\gamma, \bfr}(a)} &\ge \Abs[\big]{\Delta^{1, \bfr}(a)} - \Abs[\big]{\Delta^{\gamma, \bfr}(a) - \Delta^{1, \bfr}(a)} \\
    &\ge \delta - L (1-\gamma) > \frac{\delta}{2} > 0 ,
\end{align*}
where the first inequality is the triangle inequality, the second combines $\Abs[\big]{\Delta^{1, \bfr}(a)} \ge \delta$ with the $L$-Lipschitz bound on $(\gamma_{\text{bw}}(\bfr), 1)$, and the strict third inequality uses $1 - \gamma < \frac{\delta}{2L}$.

\smallskip\noindent
\textbf{Case $\mu^{\text{mp}, \bfr}(v) \neq \mu^{\text{mp}, \bfr}(v^\prime)$.}
Since $\bfr \in \cE_2(B)$, $\Abs[\big]{\mu^{\text{mp}, \bfr}(v^\prime) - \mu^{\text{mp}, \bfr}(v)} \ge B$. Substituting the value decomposition from \cref{lemma:value-decomposition} into the suboptimality gap yields
\[
    \Delta^{\gamma, \bfr}(a) = \bfr(v, v^\prime) - \mu^{\text{mp}, \bfr}(v^\prime) + \frac{\mu^{\text{mp}, \bfr}(v^\prime) - \mu^{\text{mp}, \bfr}(v)}{1-\gamma} + \gamma w^{\bfr}_{v^\prime}(\gamma) - w^{\bfr}_v(\gamma) .
\]
Applying the reverse triangle inequality to isolate the singular term:
\begin{align*}
    \Abs {\Delta^{\gamma, \bfr}(a)}
    &\ge \frac{\Abs {\mu^{\text{mp}, \bfr}(v^\prime) - \mu^{\text{mp}, \bfr}(v)}}{1-\gamma} - \Abs[\big]{ \bfr(v, v^\prime) - \mu^{\text{mp}, \bfr}(v^\prime) + \gamma w^{\bfr}_{v^\prime}(\gamma) - w^{\bfr}_v(\gamma) } \\
    &\ge \frac{B}{1-\gamma} - \bigl( 2M + 4nM + 4nM \bigr) \\
    &\ge \frac{B}{1-\gamma} - 10nM > 0 ,
\end{align*}
where the first inequality is the reverse triangle inequality together with $\Abs[\big]{\mu^{\text{mp}, \bfr}(v^\prime) - \mu^{\text{mp}, \bfr}(v)} \ge B$, the second uses $\Abs{\bfr(v, v^\prime)} \le M$ and $\Abs{\mu^{\text{mp}, \bfr}(v^\prime)} \le M$ (from $\cE_1(M)$) together with $\Abs{w^{\bfr}_v(\gamma)}, \Abs{w^{\bfr}_{v^\prime}(\gamma)} \le 4nM$ from \cref{lemma:value-decomposition}, the third uses $n \ge 1$, and the final strict inequality uses $1 - \gamma < \frac{B}{10nM}$.

\smallskip
Combining the two cases, $\Abs[\big]{\Delta^{\gamma, \bfr}(a)} > 0$ for every non-optimal edge $a$ and every $\gamma \in (\gamma_{\text{bw}}(\bfr), 1)$. By definition of Blackwell optimality, on $(\gamma_{\text{bw}}(\bfr), 1)$ the pair $(\pi^{\text{bw}}_{\min}, \pi^{\text{bw}}_{\max})$ is optimal, so for non-optimal edges $a$ outgoing from $v$, the sign of $\Delta^{\gamma, \bfr}(a)$ is positive if $v \in V_{\min}$ and negative if $v \in V_{\max}$. Since $\Delta^{\gamma, \bfr}(a)$ is continuous in $\gamma$ on $[0, 1)$ and its magnitude stays strictly positive throughout $(\gamma_{\text{bw}}(\bfr), 1)$, continuity extends the strictly positive magnitude and the consistent sign to a left-neighborhood $(\gamma_{\text{bw}}(\bfr) - \eta, \gamma_{\text{bw}}(\bfr)]$ for some $\eta > 0$ uniform over the finite edge set. On this neighborhood, $a^\star(v)$ is the unique edge at $v$ achieving $\Delta^{\gamma, \bfr}(a^\star(v)) = 0$, so the pair $(\pi^{\text{bw}}_{\min}, \pi^{\text{bw}}_{\max})$ remains optimal on $(\gamma_{\text{bw}}(\bfr) - \eta, 1)$. This contradicts the definition of $\gamma_{\text{bw}}(\bfr)$ as the infimum of valid thresholds. Therefore, $\gamma_{\text{bw}}(\bfr) \le \max\bigl(1 - \frac{\delta}{36 n^2 M}, 1 - \frac{B}{10 n M}\bigr)$.

For a given $x \ge 1$, we substitute the designated $M(x) = 1 + \frac{1}{\theta} \ln (m x)$, $\delta(x) \triangleq \frac{36 n^2 M(x)}{x}$, and $B(x) \triangleq \frac{10 n M(x)}{x}$.
By definition, $\frac{\delta(x)}{36 n^2 M(x)} = \frac{B(x)}{10 n M(x)} = \frac{1}{x}$, which implies $\gamma_{\text{bw}}(\bfr) \le \max\bigl(0, 1 - \frac{1}{x}\bigr)$.
Under the event $\cE_1(M(x)) \cap \cE_2(B(x)) \cap \cE_3(\delta(x)) \cap \cE_4$, we have $\gamma_{\text{bw}}(\bfr) \le 1 - \frac{1}{x}$.
Therefore, the event $\gamma_{\text{bw}}(\bfr) > 1 - \frac{1}{x}$ is upper bounded:
\[
    \Pr_{\xi} \bigl[ \gamma_{\text{bw}}(\bfr) > 1 - \tfrac{1}{x} \bigr] \le \Pr_{\xi} \bigl[ \overline{\cE_1}(M(x)) \bigr] + \Pr_{\xi} \bigl[ \overline{\cE_2}(B(x)) \bigr] + \Pr_{\xi} \bigl[ \overline{\cE_3}(\delta(x)) \bigr] + \Pr_{\xi} \bigl[ \overline{\cE_4} \bigr] .
\]
For $\Pr_{\xi} \bigl[ \overline{\cE_1}(M(x)) \bigr]$, we have
\begin{align*}
    \Pr_{\xi} \bigl[ \overline{\cE_1}(M(x)) \bigr] &\le m \exp\bigl(- \theta (M(x)-1) \bigr)\\
    &= m \exp\bigl(- \theta \frac{1}{\theta} \ln (m x) \bigr) = \frac{1}{x} ,
\end{align*}
where the first inequality is due to \cref{eq:reward-tail-bound} and a union bound over the $m$ edges, and the first equality is due to the definition of $M(x)$.
For $\Pr_{\xi} \bigl[ \overline{\cE_2}(B(x)) \bigr]$, we have
\begin{align*}
    \Pr_{\xi} \bigl[ \overline{\cE_2}(B(x)) \bigr] &\le 2 B(x) m^2 n \phi\\
    &= 2 m^2 n \frac{10 n M(x)}{x} \phi = \frac{20 m^2 n^2 \phi M(x)}{x} ,
\end{align*}
where the first inequality is due to \cref{lemma:non-ergodic-value-separation} and a union bound over the $m$ edges, and the first equality is due to the definition of $B(x)$.
For $\Pr_{\xi} \bigl[ \overline{\cE_3}(\delta(x)) \bigr]$, we have
\begin{align*}
    \Pr_{\xi} \bigl[ \overline{\cE_3}(\delta(x)) \bigr] &\le 2 m \delta(x) \phi\\
    &= 2 m \frac{36 n^2 M(x)}{x} \phi = \frac{72 m n^2 \phi M(x)}{x} ,
\end{align*}
where the first inequality is due to \cref{lemma:value-separation-at-one} and a union bound over the $m$ edges, and the first equality is due to the definition of $\delta(x)$.
We also have $\Pr_{\xi} \bigl[ \overline{\cE_4} \bigr] = 0$ by \cref{proposition:optimality-uniqueness}.
Therefore, we obtain
\begin{align*}
    \Pr_{\xi} \bigl[ \gamma_{\text{bw}}(\bfr) > 1 - \tfrac{1}{x} \bigr] &\le \frac{1}{x} + \frac{20 m^2 n^2 \phi M(x)}{x} + \frac{72 m n^2 \phi M(x)}{x}\\
    &= \frac{1 + 4 m n^2 (5m + 18) \phi M(x)}{x}\\
    &\le \frac{1 + 92 m^2 n^2 \phi M(x)}{x} \le \frac{K M(x)}{x} ,
\end{align*}
where the second inequality uses $5m + 18 \le 23 m$ (which holds since $m \ge 1$),
and the last inequality uses $1 \le 92 m^2 n^2 M(x)$ (which holds since $M(x) \ge 1$ for $x \ge 1$ and $m, n \ge 1$) and the definition of $K$ in \cref{eq:K-definition}.
This completes the proof.
\end{proof}

\subsection{Restarting Policy Iteration}\label{sec:repi-deterministic-bound}

In this subsection, we introduce the restarting variant of policy iteration called $\algname{RePI}$, presented in full detail in \cref{alg:restarting-tpi}.
We bound its total iteration count in \cref{lemma:general-repi-total-iteration-bound}: deterministically by the number of positional strategies for every reward vector, and, under the smoothed model, by an inverse-polynomial factor of the Blackwell gap $1 - \gamma_{\text{bw}}(\bfr)$ with probability~$1$.

\paragraph{Informal description.}
Our algorithm sweeps through a geometric sequence of thresholds $\gamma_k = 1 - e^{-k}$ for $k = 0, 1, \dots, \lceil n \ln n \rceil$.
For each $k$ in turn, if the target discount factor $\gamma$ already satisfies $\gamma \le \gamma_k$, the algorithm directly executes the standard $\algname{PI}$ procedure on $\gamma$.
Otherwise, it solves the game with the discount factor~$\gamma_k$ to obtain a candidate strategy pair, computes the payoff of this pair under the target discount factor $\gamma$, and checks for any unilateral improving switches.
If none exists, the candidate is optimal for the discount factor $\gamma$ and the algorithm terminates.
If the test fails, the algorithm proceeds to the next threshold $\gamma_{k+1}$ and retries.
If every threshold in the sequence fails, then the algorithm falls back to running the standard $\algname{PI}$ procedure on $\gamma$.

\begin{algorithm}[htb]
\caption{Restarting Policy Iteration: $\algname{RePI}(G, \gamma, \bfr)$}\label{alg:restarting-tpi}
\begin{algorithmic}[1]
    \State \textbf{Input:} Graph $G = (V_{\min}, V_{\max}, E)$ with $n = \Abs {V} \ge 2$, target discount factor $\gamma \in [0,1)$, reward vector $\bfr \in \R^E$.
    \State \textbf{Output:} Optimal strategy pair $(\pi_{\min}^*, \pi_{\max}^*)$ and value vector $\mu^{\gamma, \bfr}$ for the game $G^{\gamma, \bfr}$.
    \State Let $k_{\max} = \lceil n \ln n \rceil$.
    \For{$k = 0, 1, \dots, k_{\max}$}
        \State Set the threshold discount factor $\gamma_k = 1 - e^{-k}$.
        \If{$\gamma \le \gamma_k$}
            \State \Return $\algname{PI}(G, \gamma, \bfr)$
        \Else
            \State \textbf{// Phase 1: Solve at the threshold}
            \State $(\pi_{\min}, \pi_{\max}, \mu_{\gamma_k}) \gets \algname{PI}(G, \gamma_k, \bfr)$
            \State \textbf{// Phase 2: Optimality test for the target discount factor}
            \State Compute the value vector $\mu_\gamma \in \R^V$ by evaluating $(\pi_{\min}, \pi_{\max})$ under $\gamma$.
            \State Let $I_{\min} = \bigl\{ v \in V_{\min} \mid \exists (v, v^\prime) \in E \text{ such that } \bfr(v, v^\prime) + \gamma \mu_\gamma(v^\prime) < \mu_\gamma(v) \bigr\}$.
            \State Let $I_{\max} = \bigl\{ v \in V_{\max} \mid \exists (v, v^\prime) \in E \text{ such that } \bfr(v, v^\prime) + \gamma \mu_\gamma(v^\prime) > \mu_\gamma(v) \bigr\}$.
            \If{$I_{\min} = \emptyset$ \textbf{and} $I_{\max} = \emptyset$}
                \State \Return $(\pi_{\min}, \pi_{\max}, \mu_\gamma)$
            \EndIf
        \EndIf
    \EndFor
    \State \textbf{// Phase 3: Fallback}
    \State \Return $\algname{PI}(G, \gamma, \bfr)$
\end{algorithmic}
\end{algorithm}

\paragraph{Comparison with~\cite{loff2024smoothed}.}
\Cref{alg:restarting-tpi} is related in spirit to the increasing-discount algorithm $\algname{IncreasingDiscountPI2}$ of~\cite{loff2024smoothed}, but the crucial difference is that \cref{alg:restarting-tpi} has an explicit fallback after a fixed finite sequence of thresholds.
This gives a uniform deterministic bound on the running time, which is essential for our expected-runtime analysis.
Thus, although $\algname{IncreasingDiscountPI2}$ yields a high-probability smoothed guarantee in~\cite{loff2024smoothed}, the algorithm does not achieve a polynomial expected-running-time guarantee based on our techniques.
Moreover, $\algname{IncreasingDiscountPI2}$ warm-starts policy iteration from the policies computed at the previous discount factor, whereas \cref{alg:restarting-tpi} restarts policy iteration from scratch at each threshold; this difference, however, is not used in our analysis.

\begin{lemma}[Total iteration bound]\label{lemma:general-repi-total-iteration-bound}
Let $G$ be any game graph with $n \ge 2$, let $\gamma \in [0,1)$, and let $\bfr \in \R^E$. Denote by $N_{\text{iter}}(\bfr)$ the total number of iterations executed by $\algname{RePI}(G, \gamma, \bfr)$ (\cref{alg:restarting-tpi}), summed over all calls to $\algname{PI}$.
Then $N_{\text{iter}}(\bfr) \le (\lceil n \ln n \rceil + 2) \, n^n$ for every $\bfr$, and, under the smoothed model, with probability $1$,
\[
    N_{\text{iter}}(\bfr) \le 27 \cdot \frac{m}{1-\gamma_{\text{bw}}(\bfr)} \ln\Bigl( \frac{e \cdot n}{1-\gamma_{\text{bw}}(\bfr)} \Bigr) .
\]
\end{lemma}

\begin{proof}
Let $k_{\max} = \lceil n \ln n \rceil$ and let $k_{\text{final}} \in \{0, 1, \dots, k_{\max}\}$ denote the index of the last iteration of the for-loop executed by $\algname{RePI}$.
The algorithm performs $\algname{PI}$ at threshold $\gamma_k$ for each $k \in \{0, 1, \dots, k_{\text{final}}\}$ (or at the target $\gamma$ when $\gamma \leq \gamma_{k_{\text{final}}}$), and optionally one fallback $\algname{PI}(G, \gamma, \bfr)$ in Phase~3.
By \cref{prop:pi-worst-case-bound}, each $\algname{PI}$ call at threshold $\gamma_k = 1 - e^{-k}$ executes at most $\frac{6m}{1-\gamma_k} \ln\bigl(\frac{n}{1-\gamma_k}\bigr) = 6m \, e^k (\ln n + k)$ iterations, and is also trivially bounded by $n^n$ because the number of iterations is upper bounded by the number of positional strategies. Therefore, we have
\[
    N_{\text{iter}}(\bfr) \le \sum_{k=0}^{k_{\text{final}}} \min\bigl(6m \, e^k (\ln n + k), \, n^n\bigr) + n^n \cdot \1\{\text{Phase 3 triggered}\} .
\]

\smallskip\noindent\emph{Deterministic bound.}
Bounding every term in the sum by $n^n$ and using that there are at most $k_{\max}+1$ terms plus at most one fallback gives $N_{\text{iter}}(\bfr) \le (\lceil n \ln n \rceil + 2) \, n^n$ for every $\bfr$.

\smallskip\noindent\emph{Smoothed geometric bound.}
Let $\Gamma_{\text{alg}} \triangleq \{\gamma_k : k = 0,1,\dots,k_{\max}\}$ be the finite set of thresholds used by the algorithm.
By \cref{lemma:unique-threshold-optimal-blackwell}, with probability $1$, for every $\gamma_k \in \Gamma_{\text{alg}}$ with $\gamma_k > \gamma_{\text{bw}}(\bfr)$, the game $G^{\gamma_k,\bfr}$ has a unique optimal strategy pair, and this pair is Blackwell optimal.
Condition on this probability-one event.
We bound the total in terms of $\gamma_{\text{bw}}(\bfr)$, splitting into two cases according to whether the Blackwell threshold is resolved within the for-loop.
In each case we establish the two bounds
\begin{equation}\label{eq:repi-kfinal-bounds}
    e^{k_{\text{final}}+1} \le \frac{e^2}{1-\gamma_{\text{bw}}(\bfr)}
    \qquad \text{and} \qquad
    \ln n + k_{\text{final}} \le \ln\Bigl(\frac{en}{1-\gamma_{\text{bw}}(\bfr)}\Bigr) .
\end{equation}

\smallskip\noindent\textbf{Case 1: $\gamma_{\text{bw}}(\bfr) < \gamma_{k_{\max}}$.}
Since $\gamma_0 = 0$, there exists $i \in \{0, 1, \dots, k_{\max}-1\}$ with $\gamma_i \le \gamma_{\text{bw}}(\bfr) < \gamma_{i+1}$.
Because $\gamma_{i+1} > \gamma_{\text{bw}}(\bfr)$, the conditioned event implies that the unique optimal pair of $G^{\gamma_{i+1}, \bfr}$ is Blackwell optimal, so $\algname{PI}(G, \gamma_{i+1}, \bfr)$ returns a Blackwell optimal pair; being Blackwell optimal, it is optimal for every $\gamma^\prime \in [\gamma_{\text{bw}}(\bfr), 1)$.
Hence, if iteration $k = i+1$ reaches Phase~2 (that is, if $\gamma > \gamma_{i+1}$), the returned pair is optimal for the target $\gamma$ and Phase~2 succeeds; otherwise $\gamma \le \gamma_{i+1}$ and the algorithm has already returned $\algname{PI}(G, \gamma, \bfr)$ at some iteration $k \le i+1$.
Either way the algorithm terminates without triggering Phase~3, so $k_{\text{final}} \le i+1$ and the Phase~3 term vanishes.
From $1 - \gamma_i = e^{-i} \ge 1 - \gamma_{\text{bw}}(\bfr)$ we obtain $e^i \le 1/(1-\gamma_{\text{bw}}(\bfr))$, hence
\[
    e^{k_{\text{final}}+1} \le e^{i+2} \le \frac{e^2}{1-\gamma_{\text{bw}}(\bfr)}
    \qquad \text{and} \qquad
    k_{\text{final}} \le \ln\frac{1}{1-\gamma_{\text{bw}}(\bfr)} + 1 ,
\]
which gives \cref{eq:repi-kfinal-bounds} since $\ln n + k_{\text{final}} \le \ln n + \ln\frac{1}{1-\gamma_{\text{bw}}(\bfr)} + 1 = \ln\bigl(en/(1-\gamma_{\text{bw}}(\bfr))\bigr)$.

\smallskip\noindent\textbf{Case 2: $\gamma_{\text{bw}}(\bfr) \ge \gamma_{k_{\max}}$.}
Here $1/(1-\gamma_{\text{bw}}(\bfr)) \ge e^{k_{\max}} \ge n^n$, and the algorithm may run all $k_{\max}+1$ for-loop iterations plus a Phase~3 fallback, so $k_{\text{final}} = k_{\max}$. Hence
\[
    e^{k_{\text{final}}+1} = e \cdot e^{k_{\max}} \le \frac{e}{1-\gamma_{\text{bw}}(\bfr)} \le \frac{e^2}{1-\gamma_{\text{bw}}(\bfr)}
    \qquad \text{and} \qquad
    k_{\text{final}} = k_{\max} \le \ln\frac{1}{1-\gamma_{\text{bw}}(\bfr)} ,
\]
which gives \cref{eq:repi-kfinal-bounds} since $\ln n + k_{\text{final}} \le \ln n + \ln\frac{1}{1-\gamma_{\text{bw}}(\bfr)} \le \ln\bigl(en/(1-\gamma_{\text{bw}}(\bfr))\bigr)$.

\smallskip
We now bound the for-loop contribution using \cref{eq:repi-kfinal-bounds}:
\begin{align*}
    \sum_{k=0}^{k_{\text{final}}} 6m \, e^k (\ln n + k)
    &\le 6m \, (\ln n + k_{\text{final}}) \sum_{k=0}^{k_{\text{final}}} e^k \\
    &\le 6m \, (\ln n + k_{\text{final}}) \cdot \frac{e^{k_{\text{final}}+1}}{e-1} \\
    &\le \frac{6 e^2}{e-1} \cdot \frac{m}{1-\gamma_{\text{bw}}(\bfr)} \ln\Bigl(\frac{en}{1-\gamma_{\text{bw}}(\bfr)}\Bigr) ,
\end{align*}
where the first inequality uses $\ln n + k \le \ln n + k_{\text{final}}$ for every $k \le k_{\text{final}}$, the second sums the geometric series $\sum_{k=0}^{k_{\text{final}}} e^k \le e^{k_{\text{final}}+1}/(e-1)$, and the third substitutes both bounds of \cref{eq:repi-kfinal-bounds}.
It remains to account for the Phase~3 fallback. In Case~1 this term vanishes; in Case~2 it contributes at most
\[
    n^n \le \frac{1}{1-\gamma_{\text{bw}}(\bfr)} \le \frac{m}{1-\gamma_{\text{bw}}(\bfr)} \ln\Bigl(\frac{en}{1-\gamma_{\text{bw}}(\bfr)}\Bigr) ,
\]
using $m \ge 1$ and $\ln(en) \ge 1$. Adding the two contributions yields
\[
    N_{\text{iter}}(\bfr)
    \le \Bigl(\frac{6e^2}{e-1} + 1\Bigr) \cdot \frac{m}{1-\gamma_{\text{bw}}(\bfr)} \ln\Bigl(\frac{en}{1-\gamma_{\text{bw}}(\bfr)}\Bigr)
    \le 27 \cdot \frac{m}{1-\gamma_{\text{bw}}(\bfr)} \ln\Bigl(\frac{en}{1-\gamma_{\text{bw}}(\bfr)}\Bigr) .
\]
which is the second bound. As it was derived on the probability-one event from \cref{lemma:unique-threshold-optimal-blackwell}, this bound holds with probability $1$.
\end{proof}

\subsection{High Probability Guarantees}\label{sec:threshold-PI-and-high-prob-bound-ergodic}

In this subsection, we combine the total iteration bound of \cref{lemma:general-repi-total-iteration-bound} with the tail bound from \cref{sec:tail-bound-general} to show that, with probability at least $1 - \epsilon$, $\algname{RePI}$ terminates after a number of iterations polynomial in the game parameters, the density bound, the inverse of tail bound, and the inverse failure probability $1/\epsilon$.

\begin{lemma}\label{lemma:general-RePI-high-probability-convergence}
Let $G$ be any game graph. Under the smoothed model, for any $\epsilon \in (0, 1)$, define
\[
    L_\epsilon \triangleq \theta + \ln\Bigl(1 + \frac{2 m K}{\epsilon\theta}\Bigr)
    \quad \text{and} \quad
    x_\epsilon \triangleq \frac{2 K L_\epsilon}{\epsilon \theta} .
\]
With probability at least $1 - \epsilon$ over the perturbation $\bfr$, \cref{alg:restarting-tpi} executes $\algname{PI}$ for a total of at most $27 \, m \, x_\epsilon \ln(en \, x_\epsilon)$ iterations. In particular, $x_\epsilon$ is polynomial in $n, m, \phi, 1/\theta$, and $1/\epsilon$.
\end{lemma}

\begin{proof}[Proof overview]
Fix a failure probability $\epsilon \in (0,1)$.
The tail bound in \cref{lemma:general-bw-threshold-tail} shows that the inverse Blackwell gap $1/(1-\gamma_{\text{bw}}(\bfr))$ is bounded by a polynomial in $n,m,\phi,1/\theta$, and $1/\epsilon$, with probability at least $1-\epsilon$.
Conditioned on this event, substituting this bound into the total iteration bound of \cref{lemma:general-repi-total-iteration-bound} yields a polynomial bound on the total number of iterations executed by $\algname{RePI}$.
\end{proof}

\begin{proof}
Since $L_\epsilon \ge \theta$, $K \ge 1$, $m \ge 1$, and $\epsilon \in (0,1)$, we have $x_\epsilon \ge \frac{2K}{\epsilon} > 2$.

We claim that $\gamma_{\text{bw}}(\bfr) \le 1 - 1/x_\epsilon$ with probability at least $1 - \epsilon$.
Indeed, applying \cref{lemma:general-bw-threshold-tail} with $x = x_\epsilon$ yields
\[
    \Pr_{\xi} \bigl[ \gamma_{\text{bw}}(\bfr) > 1 - \tfrac{1}{x_\epsilon} \bigr]
    \le \frac{K M(x_\epsilon)}{x_\epsilon}
    = \frac{\epsilon}{2L_\epsilon}
    \Bigl(\theta + \ln(m x_\epsilon)\Bigr),
\]
where the equality uses the definitions of $M(x)$ and $x_\epsilon$.
By the definition of $L_\epsilon$,
\[
    \ln(m x_\epsilon)
    = \ln\Bigl(\frac{2 m K}{\epsilon\theta}\Bigr) + \ln L_\epsilon
    \le \ln\Bigl(1 + \frac{2 m K}{\epsilon\theta}\Bigr) + L_\epsilon
    = 2L_\epsilon - \theta ,
\]
where the inequality uses $\ln L_\epsilon \le L_\epsilon$. Consequently,
\[
    \frac{K M(x_\epsilon)}{x_\epsilon}
    \le \frac{\epsilon}{2L_\epsilon}(\theta + 2L_\epsilon - \theta) = \epsilon ,
\]
which yields $\gamma_{\text{bw}}(\bfr) \le 1 - 1/x_\epsilon$ with probability at least $1 - \epsilon$. As the geometric bound of \cref{lemma:general-repi-total-iteration-bound} holds with probability $1$, its intersection with this event still has probability at least $1 - \epsilon$; assume both hold for the rest of the proof.

The map $y \mapsto y \ln(en \, y)$ is increasing on $y \ge 1$. Therefore, we have
\[
    \frac{1}{1-\gamma_{\text{bw}}(\bfr)} \ln\bigl(\frac{en}{1-\gamma_{\text{bw}}(\bfr)}\bigr) \le x_\epsilon \ln(en \, x_\epsilon) .
\]
\Cref{lemma:general-repi-total-iteration-bound} then yields $N_{\text{iter}}(\bfr) \le 27 \, m \, x_\epsilon \ln(en \, x_\epsilon)$.
\end{proof}

\begin{theorem}\label{theorem:general-discounted-sum-high-probability}
The restarting PI algorithm achieves smoothed polynomial time with high probability guarantee for discounted-sum games.
\end{theorem}

\begin{proof}
Fix a failure probability $\epsilon \in (0,1)$.
By \cref{lemma:general-RePI-high-probability-convergence}, with probability at least $1-\epsilon$, $\algname{RePI}$ executes at most $27 \, m \, x_\epsilon \ln(en \, x_\epsilon)$ iterations, where $x_\epsilon$ is polynomial in $n, m, \phi, 1/\theta$, and $1/\epsilon$.
\end{proof}

\subsection{Expected Runtime Guarantees}\label{sec:expected-running-time-ergodic}

To bound the expected number of iterations of \algname{RePI}, we first define the notion of truncated inverse gap and show that its expectation is bounded by a polynomial in the game parameters, the density bound, and the inverse of the tail bound in \cref{lemma:general-expected-truncated-bw} (Ingredient~4 from \cref{sec:techniques-overview}).
We then combine this bound with the total iteration bound of \cref{lemma:general-repi-total-iteration-bound} to obtain the desired polynomial bound on the expected number of iterations.

\paragraph{Truncated inverse gap.}
For any reward vector $\bfr \in \R^E$, define
\[
    Y(\bfr) \triangleq \min\left(
        1/(1-\gamma_{\text{bw}}(\bfr))
        \ln\left(1/(1-\gamma_{\text{bw}}(\bfr))\right),
        \, n^n
    \right).
\]

\begin{lemma}[Bounding expected truncated inverse gap]\label{lemma:general-expected-truncated-bw}
Let $G$ be any game graph with $n \geq 2$. Let $K$ be defined in \cref{eq:K-definition}. Under the smoothed model, we have
\[
    \mathbb{E}_{\xi} \bigl[ Y(\bfr) \bigr]
    \le \frac{6 (\theta + 1)}{\theta} K (n \ln m)^3.
\]
\end{lemma}

\begin{proof}[Proof overview]
By \cref{lemma:general-bw-threshold-tail}, the inverse Blackwell gap has tail probability of $\cO(\ln(x)/x)$.
Integrating this tail up to the truncation point $n^n$ yields the stated polynomial bound.
\end{proof}

\begin{proof}
Let $x_0$ be the unique real number satisfying $x_0 \ge 1$ and $x_0 \ln x_0 = n^n$.
We compute the expectation of $Y(\bfr)$ by integrating the tail probability and applying the substitution $y = x \ln x$, where $dy = (1 + \ln x) dx$:
\begin{equation}\label{eq:general-Y-expectation-sum-bound}
\begin{aligned}
    \mathbb{E}_{\xi} \bigl[ Y(\bfr) \bigr] &= \int_{0}^{n^n} \Pr_{\xi} \bigl[ Y(\bfr) > y \bigr] \, dy \\
    &= \int_{1}^{x_0} \Pr_{\xi} \bigl[ \gamma_{\text{bw}}(\bfr) > 1 - \tfrac{1}{x} \bigr] (1 + \ln x) \, dx \\
    &\le \int_{1}^{e} 1 \cdot (1 + \ln x) \, dx + \int_{e}^{x_0} \Pr_{\xi} \bigl[ \gamma_{\text{bw}}(\bfr) > 1 - \tfrac{1}{x} \bigr] (1 + \ln x) \, dx \\
    &= e + \int_{e}^{x_0} \Pr_{\xi} \bigl[ \gamma_{\text{bw}}(\bfr) > 1 - \tfrac{1}{x} \bigr] (1 + \ln x) \, dx .
\end{aligned}
\end{equation}
For the second integral, we substitute the bound $\Pr_{\xi} [\gamma_{\text{bw}}(\bfr) > 1 - \tfrac{1}{x}] \le \frac{K M(x)}{x}$ established in \cref{lemma:general-bw-threshold-tail}, which applies on $[e, x_0]$ since $x \ge e > 1$:
\[
    \int_{e}^{x_0} \Pr_{\xi} \bigl[ \gamma_{\text{bw}}(\bfr) > 1 - \tfrac{1}{x} \bigr] (1 + \ln x) \, dx \le K \int_{e}^{x_0} \frac{M(x)(1+\ln x)}{x} \, dx .
\]
We apply the substitution $u = \ln x$, meaning $dx/x = du$. The integration limits shift to $u \in [1, \ln x_0]$. Since $x_0 \le n^n$, the upper limit is strictly bounded by $n \ln n$. This yields:
\begin{align*}
    \int_{e}^{x_0} \frac{M(x)(1+\ln x)}{x} \, dx
    &\le \int_{1}^{n \ln n} M(e^u) (1 + u) \, du \\
    &\le M(e^{n \ln n}) \int_{1}^{n \ln n} (1 + u) \, du \\
    &= \Bigl( 1 + \frac{\ln m + n \ln n}{\theta} \Bigr) \bigl[ u + \frac{u^2}{2} \bigr]_{1}^{n \ln n} \\
    &\le \Bigl( 1 + \frac{\ln m + n \ln n}{\theta} \Bigr) \bigl( \frac{1}{2} (n \ln n)^2 + n \ln n \bigr) \\
    &\le \Bigl( 1 + \frac{\ln m + n \ln n}{\theta} \Bigr) \frac{3}{2} (n \ln n)^2 \\
    &\le \frac{\theta + 1}{\theta} (\ln m + n \ln n) \frac{3}{2} (n \ln n)^2 \\
    &\le \frac{\theta + 1}{\theta} \, 2 n \ln m \, \frac{3}{2} (n \ln n)^2 \\
    &\le \frac{\theta + 1}{\theta} \, 2 n \ln m \, \frac{3}{2} (n \ln m)^2 \\
    &= \frac{3 (\theta + 1)}{\theta} (n \ln m)^3 ,
\end{align*}
where the first equality follows from $M(e^u) = 1 + \frac{\ln m + u}{\theta}$;
the inequality $\frac{1}{2} (n \ln n)^2 + n \ln n \le \frac{3}{2} (n \ln n)^2$ uses $n \ln n \le (n \ln n)^2$ (which holds for all $n \ge 1$, since $n \ln n \ge 0$);
the bound $1 + \frac{\ln m + n \ln n}{\theta} \le \frac{\theta + 1}{\theta} (\ln m + n \ln n)$ uses $1 \le \ln m + n \ln n$ (which holds for $n \ge 2$);
the bounds $\ln m + n \ln n \le 2 n \ln m$ and $(n \ln n)^2 \le (n \ln m)^2$ both use $m \ge n$.

Substituting this bound back into \cref{eq:general-Y-expectation-sum-bound} yields $\mathbb{E}_{\xi}[Y(\bfr)] \le e + K \frac{3 (\theta + 1)}{\theta} (n \ln m)^3$.
Finally, since $\frac{\theta + 1}{\theta} \ge 1$ and $n \ln m \ge 2 \ln 2 > 1$ (as $m \ge n \ge 2$), we have $\frac{3 (\theta + 1)}{\theta} (n \ln m)^3 \ge 3 (n \ln m)^3 > 3 > e$; since $K \ge 1$, this gives $e \le K \frac{3 (\theta + 1)}{\theta} (n \ln m)^3$, and hence $\mathbb{E}_{\xi}[Y(\bfr)] \le 2 K \frac{3 (\theta + 1)}{\theta} (n \ln m)^3$, completing the proof.
\end{proof}

\begin{lemma}\label{lemma:general-repi-expected-running-time}
Let $G$ be any game graph with $n \ge 2$.
Under the smoothed model, the expected total number of iterations in all the \algname{PI}'s executed by \cref{alg:restarting-tpi} is bounded by
\[
    324 \, m \ln n \, K \frac{3 (\theta + 1)}{\theta} (n \ln m)^3 ,
\]
where $K$ is defined in \cref{eq:K-definition}. In particular, this bound is polynomial in $n, m$, $\phi$, and $1/\theta$.
\end{lemma}

\begin{proof}[Proof overview]
By \cref{lemma:general-repi-total-iteration-bound}, the runtime of \algname{RePI} is bounded by the inverse Blackwell gap (up to polynomial factors of $n$ and $m$), while its worst-case fallback is bounded by the finite number of strategies.
These two bounds are captured by the truncated inverse gap $Y(\bfr)$.
Therefore, the polynomial expectation of $Y(\bfr)$ from \cref{lemma:general-expected-truncated-bw} yields a polynomial expected runtime.
\end{proof}

\begin{proof}
Let $x \triangleq 1/(1-\gamma_{\text{bw}}(\bfr)) \ge 1$, and recall the truncated inverse gap $Y(\bfr) = \min\bigl(x \ln x, \, n^n\bigr)$.
By \cref{lemma:general-repi-total-iteration-bound}, with probability $1$ we have
\[
    N_{\text{iter}}(\bfr) \le \min\Bigl( 27 \, m \, x \ln(en \, x), \ (\lceil n \ln n \rceil + 2) \, n^n \Bigr) .
\]
We claim that
\begin{equation}\label{eq:repi-expected-pointwise}
    N_{\text{iter}}(\bfr) \le 27 \, m \bigl( (\ln n + 2) \, Y(\bfr) + e(\ln n + 1) \bigr) .
\end{equation}
Indeed, we prove this by splitting into two cases according to the truncation.

\smallskip\noindent\textbf{Case 1: $x \ln x \le n^n$.}
Here $Y(\bfr) = x \ln x$, and we bound the minimum by its first term:
\begin{align*}
    N_{\text{iter}}(\bfr)
    &\le 27 \, m \, x \ln(en \, x)\\
    &= 27 \, m \bigl( x \ln x + (\ln n + 1) x \bigr) \\
    &\le 27 \, m \bigl( Y(\bfr) + (\ln n + 1)(Y(\bfr) + e) \bigr)\\
    &= 27 \, m \bigl( (\ln n + 2) Y(\bfr) + e(\ln n + 1) \bigr) ,
\end{align*}
where the first equality uses $\ln(en\,x) = \ln x + \ln n + 1$, and the second inequality uses $x \ln x = Y(\bfr)$ together with $x \le x \ln x + e = Y(\bfr) + e$, valid for all $x \ge 1$.

\smallskip\noindent\textbf{Case 2: $x \ln x > n^n$.}
Here $Y(\bfr) = n^n$, and we bound the minimum by its second term. Since every vertex has at least one outgoing edge, $m \ge n$, hence $\lceil n \ln n \rceil + 2 \le n \ln n + 3 \le 27 \, m (\ln n + 2)$. Therefore
\[
    N_{\text{iter}}(\bfr)
    \le (\lceil n \ln n \rceil + 2) \, n^n
    \le 27 \, m (\ln n + 2) \, n^n
    \le 27 \, m \bigl( (\ln n + 2) Y(\bfr) + e(\ln n + 1) \bigr) .
\]

\smallskip
Consequently, \cref{eq:repi-expected-pointwise} holds with probability $1$; since $N_{\text{iter}}(\bfr) \le (\lceil n \ln n \rceil + 2) \, n^n$ for every $\bfr$ by \cref{lemma:general-repi-total-iteration-bound}, the complementary null event contributes nothing to $\mathbb{E}_{\xi}[N_{\text{iter}}(\bfr)]$.
Taking the expectation of \cref{eq:repi-expected-pointwise} over the smoothed perturbation $\bfr$ and applying \cref{lemma:general-expected-truncated-bw},
\begin{align*}
    \mathbb{E}_{\xi}[N_{\text{iter}}(\bfr)]
    &\le 27 \, m \bigl( (\ln n + 2) \, \mathbb{E}_{\xi}[Y(\bfr)] + e(\ln n + 1) \bigr)\\
    &\le 27 \, m \Bigl( 2 (\ln n + 2) K \frac{3 (\theta + 1)}{\theta} (n \ln m)^3 + e(\ln n + 1) \Bigr)\\
    &\le 81 \, m (\ln n + 2) K \frac{3 (\theta + 1)}{\theta} (n \ln m)^3\\
    &\le 324 \, m \ln n \, K \frac{3 (\theta + 1)}{\theta} (n \ln m)^3 ,
\end{align*}
where the third inequality uses $e (\ln n + 1) \le (\ln n + 2) K \frac{3 (\theta + 1)}{\theta} (n \ln m)^3$, which follows from $e \le K \frac{3 (\theta + 1)}{\theta} (n \ln m)^3$ (since $K \ge 1$ and $\frac{3 (\theta + 1)}{\theta} (n \ln m)^3 \ge 3 (n \ln m)^3 \ge 3 > e$), and the last inequality uses $\ln n + 2 \le 4 \ln n$ (which holds for $n \ge 2$). This is the claimed bound.
\end{proof}

\begin{theorem}\label{theorem:general-discounted-sum-expected}
The restarting PI algorithm achieves smoothed polynomial time with expected runtime guarantee for discounted-sum games.
\end{theorem}

\begin{proof}
By \cref{lemma:general-repi-expected-running-time}, the expected total number of iterations executed by $\algname{RePI}$ is at most $324 \, m \ln n \, K \frac{3 (\theta + 1)}{\theta} (n \ln m)^3$, which is polynomial in $n, m, \phi$, and $1/\theta$ by the definition of $K$ in \cref{eq:K-definition}.
\end{proof}

\section{Mean-payoff games}\label{sec:mean-payoff-games}
As established in \cref{sec:background}, any Blackwell optimal strategy for a discounted-sum game is an optimal strategy for the corresponding mean-payoff game~\citep{bewley1976asymptotic}.
Therefore, our smoothed analysis for computing Blackwell optimal strategies directly yields a smoothed polynomial-time algorithm for mean-payoff games, presented in \cref{theorem:mpg-main}.

\begin{corollary}\label{theorem:mpg-main}
The restarting mean-payoff PI algorithm $\algname{RePI}_{\text{mp}}$ achieves smoothed polynomial time for mean-payoff games, with both high-probability and expected runtime guarantees.
\end{corollary}

In the rest of this section, we provide the technical development for mean-payoff games that establishes \cref{theorem:mpg-main}, combining a high-probability guarantee (\cref{lemma:mpg-high-probability-convergence}) and an expected runtime guarantee (\cref{lemma:mpg-repi-expected-running-time}).
We first adapt the restarting algorithm of \cref{sec:repi-deterministic-bound} to the mean-payoff setting and establish a total iteration bound for it (\cref{sec:repi-mp-deterministic-bound}).
Combining this bound with the tail bound of \cref{sec:tail-bound-general} yields smoothed polynomial runtime with high probability (\cref{sec:high-prob-bound-mpg}); combining it with the expected truncated inverse gap of \cref{sec:expected-running-time-ergodic} yields smoothed polynomial expected runtime (\cref{sec:expected-running-time-mpg}).

\subsection{Restarting Mean-payoff Policy Iteration}\label{sec:repi-mp-deterministic-bound}

In this subsection, we adapt the restarting variant of policy iteration to the mean-payoff setting, called $\algname{RePI}_{\text{mp}}$, presented in full detail in \cref{alg:restarting-tpi-mp}, and then establish a deterministic bound on its total iteration count in terms of the Blackwell threshold $\gamma_{\text{bw}}(\bfr)$ (\cref{lemma:mpg-repi-total-iteration-bound}).

\paragraph{Informal description.}
Our algorithm sweeps through the same geometric sequence of thresholds $\gamma_k = 1 - e^{-k}$ for $k = 0, 1, \dots, \lceil n \ln n \rceil$, with each successive threshold $e$ times closer to $1$ than the previous one.
For each $k$ in turn, it solves the game with the discount factor~$\gamma_k$ via the standard discounted-sum $\algname{PI}(G, \gamma_k, \bfr)$ to obtain a candidate strategy pair $(\pi_{\min}, \pi_{\max})$.
The algorithm then tests whether this candidate pair is already optimal for the mean-payoff game.
For this purpose, we use the combinatorial strategy evaluation of~\cite{bjorklund2007combinatorial}, which reduces mean-payoff games to the longest shortest paths problem and evaluates strategies using shortest distances to a sink.
Their method can determine whether a given positional strategy pair is mean-payoff optimal; we invoke this test as a black box.
If the test succeeds, the algorithm computes the mean-payoff value vector induced by the candidate pair and returns both.
Otherwise it proceeds to the next threshold $\gamma_{k+1}$ and retries.
If every threshold in the sequence fails to be optimal, the algorithm falls back to the longest-shortest-path strategy-improvement algorithm of~\cite{bjorklund2007combinatorial}, denoted by $\algname{PI}_{\text{mp}}(G,\bfr)$.
We use this algorithm as a deterministic black-box fallback: it returns a mean-payoff optimal positional strategy pair and value vector, and for our purposes the upper bound of at most $n^n$ strategy improvements is sufficient,
which is due to the fact that there are at most $n^n$ positional strategies pairs in a game with $n$ vertices.

\begin{algorithm}[htb]
\caption{Restarting Mean-payoff Policy Iteration: $\algname{RePI}_{\text{mp}}(G, \bfr)$}\label{alg:restarting-tpi-mp}
\begin{algorithmic}[1]
    \State \textbf{Input:} Graph $G = (V_{\min}, V_{\max}, E)$ with $n = \Abs{V} \ge 2$, reward vector $\bfr \in \R^E$.
    \State \textbf{Output:} Mean-payoff optimal strategy pair $(\pi_{\min}^*, \pi_{\max}^*)$ and mean-payoff value vector $\mu^{\text{mp}, \bfr}$ for the game $G^{\text{mp}, \bfr}$.
    \State Let $k_{\max} = \lceil n \ln n \rceil$.
    \For{$k = 0, 1, \dots, k_{\max}$}
        \State Set the threshold discount factor $\gamma_k = 1 - e^{-k}$.
        \State \textbf{// Phase 1: Solve the discounted game at the threshold}
        \State $(\pi_{\min}, \pi_{\max}, \mu_{\gamma_k}) \gets \algname{PI}(G, \gamma_k, \bfr)$
        \State \textbf{// Phase 2: Mean-payoff optimality test}
        \State Use the longest shortest paths algorithm of~\cite{bjorklund2007combinatorial} to test whether $(\pi_{\min},\pi_{\max})$ is mean-payoff optimal.
        \If{$(\pi_{\min},\pi_{\max})$ is mean-payoff optimal}
            \State Compute the mean-payoff value vector $\mu^{\mathrm{mp},\bfr}$ induced by $(\pi_{\min},\pi_{\max})$.
            \State \Return $(\pi_{\min}, \pi_{\max}, \mu^{\mathrm{mp},\bfr})$
        \EndIf
    \EndFor
    \State \textbf{// Phase 3: Fallback}
    \State \Return $\algname{PI}_{\text{mp}}(G, \bfr)$
\end{algorithmic}
\end{algorithm}

The intuition mirrors the discounted-sum case.
Under the smoothed model, with high probability some threshold $\gamma_k$ in the sequence strictly exceeds the Blackwell threshold $\gamma_{\text{bw}}(\bfr)$. Thus, the strategy pair returned by $\algname{PI}(G, \gamma_k, \bfr)$ is Blackwell optimal and consequently mean-payoff optimal~\citep{bewley1976asymptotic}, so the Phase~2 test succeeds and the algorithm terminates without reaching the fallback.

\begin{lemma}[Total iteration bound]\label{lemma:mpg-repi-total-iteration-bound}
Let $G$ be any game graph with $n \ge 2$ and let $\bfr \in \R^E$; let $N_{\text{iter}}(\bfr)$ be the total number of iterations performed by $\algname{RePI}_{\text{mp}}(G, \bfr)$ (\cref{alg:restarting-tpi-mp}), across all calls to $\algname{PI}$ and the optional $\algname{PI}_{\text{mp}}$ fallback.
Then $N_{\text{iter}}(\bfr) \le (\lceil n \ln n \rceil + 2) \, n^n$ for every $\bfr$, and, under the smoothed model, with probability $1$,
\[
    N_{\text{iter}}(\bfr) \le 27 \cdot \frac{m}{1-\gamma_{\text{bw}}(\bfr)} \ln\Bigl( \frac{en}{1-\gamma_{\text{bw}}(\bfr)} \Bigr) .
\]
\end{lemma}

\begin{proof}
The deterministic bound $N_{\text{iter}}(\bfr) \le (\lceil n \ln n \rceil + 2) \, n^n$ holds exactly as in \cref{lemma:general-repi-total-iteration-bound}: there are at most $\lceil n \ln n \rceil + 2$ calls, each terminating in at most $n^n$ iterations (the fallback $\algname{PI}_{\text{mp}}(G, \bfr)$ likewise, by~\cite{bjorklund2007combinatorial}).
For the geometric bound, condition on the probability-one event from \cref{lemma:unique-threshold-optimal-blackwell} with $\Gamma_{\text{alg}} = \{\gamma_k : k = 0,1,\dots,k_{\max}\}$, and run the same case analysis as in \cref{lemma:general-repi-total-iteration-bound} with two minor changes.
First, when $\gamma_{i+1} > \gamma_{\text{bw}}(\bfr)$, the unique optimal pair of $G^{\gamma_{i+1}, \bfr}$ is Blackwell optimal and is also \emph{mean-payoff} optimal~\citep{bewley1976asymptotic}, so $\algname{PI}(G, \gamma_{i+1}, \bfr)$ returns it; the test based on longest shortest paths from~\cite{bjorklund2007combinatorial} therefore certifies the pair as mean-payoff optimal in Phase~2, and the algorithm terminates without triggering Phase~3.
Second, when Phase~3 is triggered (i.e., $\gamma_{\text{bw}}(\bfr) \ge \gamma_{k_{\max}}$), the fallback is $\algname{PI}_{\text{mp}}(G, \bfr)$ rather than $\algname{PI}(G, \gamma, \bfr)$; it terminates in at most $n^n$ iterations~\cite{bjorklund2007combinatorial}, and the same bound $n^n \le \frac{m}{1-\gamma_{\text{bw}}(\bfr)} \ln\bigl(\frac{en}{1-\gamma_{\text{bw}}(\bfr)}\bigr)$ applies in this regime.
The geometric sum then combines as in \cref{lemma:general-repi-total-iteration-bound} to give the second bound, which holds with probability $1$ by \cref{lemma:unique-threshold-optimal-blackwell}.
\end{proof}

\subsection{High Probability Guarantees}\label{sec:high-prob-bound-mpg}

In this subsection, we combine the total iteration bound of \cref{lemma:mpg-repi-total-iteration-bound} with the tail bound from \cref{sec:tail-bound-general} to obtain a high-probability bound on the number of iterations of $\algname{RePI}_{\text{mp}}$.

\begin{lemma}\label{lemma:mpg-high-probability-convergence}
Let $G$ be any game graph. Under the smoothed model, for any $\epsilon \in (0, 1)$, let $x_\epsilon$ be as defined in \cref{lemma:general-RePI-high-probability-convergence}. With probability at least $1 - \epsilon$ over the perturbation $\bfr$, $\algname{RePI}_{\text{mp}}$ (\cref{alg:restarting-tpi-mp}) executes a total of at most $27 \, m \, x_\epsilon \ln(en \, x_\epsilon)$ iterations.
\end{lemma}

\begin{proof}
By the proof of \cref{lemma:general-RePI-high-probability-convergence}, $\gamma_{\text{bw}}(\bfr) \le 1 - 1/x_\epsilon$ with probability at least $1 - \epsilon$.
As the geometric bound of \cref{lemma:mpg-repi-total-iteration-bound} holds with probability $1$, assume this event together with it (their intersection still has probability at least $1 - \epsilon$), so $\frac{1}{1-\gamma_{\text{bw}}(\bfr)} \le x_\epsilon$, and the map $y \mapsto y \ln(en \, y)$ being increasing on $y \ge 1$ gives $\frac{1}{1-\gamma_{\text{bw}}(\bfr)} \ln\bigl(\frac{en}{1-\gamma_{\text{bw}}(\bfr)}\bigr) \le x_\epsilon \ln(en \, x_\epsilon)$.
\Cref{lemma:mpg-repi-total-iteration-bound} then yields $N_{\text{iter}}(\bfr) \le 27 \, m \, x_\epsilon \ln(en \, x_\epsilon)$.
\end{proof}

\subsection{Expected Runtime Guarantees}\label{sec:expected-running-time-mpg}

In this subsection, we combine the total iteration bound of \cref{lemma:mpg-repi-total-iteration-bound} with the bound on the expected truncated inverse gap from \cref{lemma:general-expected-truncated-bw} to obtain a bound on the expected number of iterations of $\algname{RePI}_{\text{mp}}$.

\begin{lemma}\label{lemma:mpg-repi-expected-running-time}
Let $G$ be any game graph with $n \ge 2$.
Under the smoothed model, the expected total number of iterations executed by \cref{alg:restarting-tpi-mp} (across all \algname{PI} executions and the fallback $\algname{PI}_{\text{mp}}$, when triggered) is bounded by
\[
    324 \, m \ln n \, K \frac{3 (\theta + 1)}{\theta} (n \ln m)^3 ,
\]
where $K$ is defined in \cref{eq:K-definition}. In particular, this bound is polynomial in $n, m$, $\phi$, and $1/\theta$.
\end{lemma}

\begin{proof}
The argument is identical to the proof of \cref{lemma:general-repi-expected-running-time}, invoking the total iteration bound of \cref{lemma:mpg-repi-total-iteration-bound} in place of \cref{lemma:general-repi-total-iteration-bound}.
With probability $1$ the same pointwise bound $N_{\text{iter}}(\bfr) \le 27 \, m \bigl( (\ln n + 2) Y(\bfr) + e(\ln n + 1) \bigr)$ holds, and since $N_{\text{iter}}(\bfr) \le (\lceil n \ln n \rceil + 2) n^n$ for every $\bfr$ the complementary null event contributes nothing to the expectation; taking the expectation and applying \cref{lemma:general-expected-truncated-bw} yields the stated bound.
\end{proof}

\begin{proof}[Proof of \cref{theorem:mpg-main}]
The high-probability guarantee follows from \cref{lemma:mpg-high-probability-convergence}.
The expected runtime guarantee follows from \cref{lemma:mpg-repi-expected-running-time}.
\end{proof}

\clearpage

\section*{Acknowledgement}
The authors are thankful for the early discussion with Raimundo Saona and Elahe Tohidi.
This paper receives funding from ERC CoG 863818 (ForM-SMArt) and Austrian Science Fund (FWF) 10.55776/COE12.

\bibliographystyle{alpha}
{\small
\bibliography{cite}
}

\end{document}